\documentclass[10pt,journal,compsoc]{IEEEtran}

\usepackage[utf8]{inputenc} 
\usepackage[T1]{fontenc}    
\usepackage{hyperref}       
\usepackage{url}            
\usepackage{booktabs}       
\usepackage{amsfonts}       
\usepackage{nicefrac}       
\usepackage{microtype}      
\usepackage{xcolor}         
\usepackage{graphicx}
\usepackage{amsthm}
\usepackage{amsmath}

\usepackage{amssymb}
\usepackage{siunitx}
\usepackage{stfloats}
\usepackage{subfigure}
\usepackage{hyperref}
\usepackage{enumerate}
\usepackage{ulem}
\usepackage{enumitem}
\usepackage{xcolor}

\usepackage{ulem} 
\usepackage{bm}
\usepackage{nicefrac}
\usepackage{booktabs}
\usepackage{array}
\usepackage{multirow}
\usepackage{threeparttable}
\usepackage{makecell}
\usepackage[procnumbered,ruled,vlined,linesnumbered]{algorithm2e}
\usepackage{caption}

\def\sizeof#1{\left|#1  \right|}

\def\eps{\epsilon}

\def\trace#1{\mathrm{Tr} \left(#1 \right)}
\def\norm#1{\left\| #1 \right\|}

\newtheorem{problem}{Problem}
\newtheorem{theorem}{Theorem}[section]

\newtheorem{lemma}[theorem]{Lemma}

\newtheorem{definition}[theorem]{Definition}

\def\calC{\mathcal{C}}
\def\calG{\mathcal{G}}

\def\calS{\mathcal{S}}

\def\calH{\mathcal{H}}
\def\calI{\mathcal{I}}

\def\norm#1{\left\| #1 \right\|}

\def\kh#1{\left( #1 \right)}

\newcommand{\removelatexerror}{\let\@latex@error\@gobble}

\newcommand{\rea}{\mathbb{R}}
\newcommand{\LaplSolver}{\textsc{Solver}\xspace}

\newcommand{\InforCenMin}{\textsc{InforCenMin}\xspace}

\newcommand\LL{\bm{\mathit{L}}}

\def\trace#1{\mathrm{Tr} \left(#1 \right)}
\def\sizeof#1{\left|#1  \right|}

\def\aa{\pmb{\mathit{a}}}

\newcommand{\ExactSM}{\textsc{ExactSM}\xspace}

\newcommand{\ApproxiSC}{\textsc{ApproxiSC}\xspace}
\newcommand{\FastSC}{\textsc{FastICM}\xspace}

\newcommand\xx{\boldsymbol{\mathit{x}}}
\newcommand\aaa{\boldsymbol{\mathit{a}}}

\newcommand\bb{\boldsymbol{\mathit{b}}}
\newcommand\cc{\boldsymbol{\mathit{c}}}
\newcommand\dd{\boldsymbol{\mathit{d}}}
\newcommand\ee{\boldsymbol{\mathit{e}}}

\renewcommand\AA{\boldsymbol{\mathit{A}}}
\newcommand\BB{\boldsymbol{\mathit{B}}}
\newcommand\CC{\boldsymbol{\mathit{C}}}
\newcommand\JJ{\boldsymbol{\mathit{J}}}
\newcommand\DD{\boldsymbol{\mathit{D}}}

\newcommand\PP{\boldsymbol{\mathit{P}}}

\newcommand{\InforCenMinD}{\textsc{InforCenMinD}\xspace}
\newcommand{\lea}{v}
\newcommand{\del}{P}
\newcommand{\SC}{\calS}
\newcommand{\samv}{{\tilde{V}}}
\DeclareMathOperator*{\argmin}{arg\,min}
\DeclareMathOperator*{\argmax}{arg\,max}

\DontPrintSemicolon
\SetKw{KwAnd}{and}
\SetFuncSty{textsc}
\SetKwInOut{Input}{Input\ \ \ \ }
\SetKwInOut{Output}{Output}

\ifCLASSOPTIONcompsoc
  \usepackage[nocompress]{cite}
\else
  \usepackage{cite}
\fi

\ifCLASSINFOpdf

\else

\fi

\begin{document}
\title{A Fast Algorithm for Moderating Critical Nodes via Edge Removal}

\author{Changan~Liu, 
        Xiaotian~Zhou, 
        Ahad N.~Zehmakan, 
        and~Zhongzhi~Zhang,~\IEEEmembership{Member,~IEEE}

\IEEEcompsocitemizethanks{
\IEEEcompsocthanksitem This work was supported by the National Natural Science Foundation of China (No. U20B2051 and No. 62372112). \textit{(Corresponding author: Zhongzhi~Zhang.)} 

\IEEEcompsocthanksitem Changan~Liu, Xiaotian~Zhou,  and Zhongzhi Zhang  are with Shanghai Key Laboratory of Intelligent Information Processing, School of Computer Science, Fudan University, Shanghai 200433, China. Changan~Liu is also with Center for Complex Decision Analysis, Fudan University, Shanghai 200433, China.
\protect\\
E-mail: 19110240031@fudan.edu.cn, 20210240043@fudan.edu.cn,
zhangzz@fudan.edu.cn 
\IEEEcompsocthanksitem Ahad N. Zehmakan is with the School of Computing, the Australian National University, Canberra, Australia. 
\protect\\
E-mail: ahadn.zehmakan@anu.edu.au
}
\thanks{Manuscript received xxxx; revised xxxx.}
}

\markboth{IEEE Transactions on Knowledge and Data Engineering,~Vol.~xx, No.~xx, May~2023}%
{Shell \MakeLowercase{\textit{et al.}}: Bare Demo of IEEEtran.cls for Computer Society Journals}

\IEEEtitleabstractindextext{%
\begin{abstract}
Critical nodes in networks are extremely vulnerable to malicious attacks to trigger negative cascading events such as the spread of misinformation and diseases. Therefore, effective moderation of critical nodes is very vital for mitigating the potential damages caused by such malicious diffusions. The current moderation methods are computationally expensive. Furthermore, they disregard the fundamental metric of information centrality, which measures the dissemination power of nodes.

We investigate the problem of removing $k$ edges from a network to minimize the information centrality of a target node $\lea$ while preserving the network's connectivity. We prove that this problem is computationally challenging: it is NP-complete and its objective function is not supermodular. However, we propose three approximation greedy algorithms using novel techniques such as random walk-based Schur complement approximation and fast sum estimation. One of our algorithms runs in nearly linear time in the number of edges.

To complement our theoretical analysis, we conduct a comprehensive set of experiments on synthetic and real networks with over one million nodes. Across various settings, the experimental results illustrate the effectiveness and efficiency of our proposed algorithms.
\end{abstract}

}

\maketitle

\IEEEdisplaynontitleabstractindextext

\IEEEpeerreviewmaketitle

\IEEEraisesectionheading{\section{Introduction}\label{sec:introduction}}

\IEEEPARstart{A} {} broad range of dynamic processes on graphs have been analyzed to attain a more profound comprehension of diverse real-world phenomena, such as the spread of misinformation on online social media platforms, the proliferation of computer viruses over the internet, and the dissemination of diseases among individuals~\cite{songced2021,zareie2022rumour,wang2023lightweight,freitas2022graphrobustness}. As a result, there has been a burgeoning interest in investigating the influence of the underlying graph structure on various characteristics of these dynamic processes~\cite{bertagnolli2021quantifying,sun2023scalable,ren2018dismantling,YiZhPa20}. Specifically, a considerable amount of attention has been focused on comprehending to what extent certain objectives can be achieved through the manipulation of the network structure. Examples of such manipulation strategies comprise eliminating nodes (such as blocking an account on an online social platform or administering a vaccine to an individual), adding edges (such as link recommendation in online social networks or constructing a physical link between two routers), or removing edges (such as restricting two individuals from meeting through quarantine measures or not exposing the posts from one user to another in an online platform)~\cite{ren2018dismantling,shan2018improve,EnMoBr12,MaKa09}. 
Furthermore, the intended objective can vary widely, ranging from minimizing the number of nodes infected by a virus to maximizing the fairness of workload among routers across different internet service providers, or reducing the proportion of users exposed to misinformation~\cite{zareie2022rumour,MaKa09,yan2019rumor,freitas2022graphrobustness,tsioutsiouliklis2022link}.

The real-world networks exhibit heterogeneous nature with critical nodes being far more essential in network structure and function~\cite{LU20161,ren2018dismantling}, such as information diffusion and system stability~\cite{ren2018dismantling,hofmann2015leadership,fan2020finder}. This confers an unbridled degree of power to such nodes, which could potentially result in significant financial, sociological, and political damages. For instance, subsequent to the hack of The Associated Press Twitter account, a false report was disseminated claiming that ``Breaking: Two Explosions in the White House and Barack Obama is injured''. This rumor resulted in losses of 10 billion USD within a few hours and caused the United States stock market to crash within minutes~\cite{peter2013bogus}. As another example, infectious diseases cause 10 million deaths each year globally, accounting for 23\% of the total disease related deaths, where critical nodes play a massive role in the extent of the spread~\cite{nandi2016methods}. Consequently, there has been an increasing interest in shedding light on managing and alleviating the impact of a set of target nodes, especially influential nodes~\cite{wang2013negative,taninmics2020minimizing,khalil2014scalable,yao2015minimizing}. 

Another application space where the problem of mitigating some target nodes have gained substantial attention is the realm of privacy protection in networks. In this setup, the goal is to protect the privacy of users or conceal critical entities in the network by implementing structural anonymization~\cite{waniekHidingIndividualsCommunities2018, ji2019greedily}. Please refer to  Sections~\ref{sec:related01} and~\ref{sec:related02} for more details on these topics when appropriate.

Two commonly employed graph operations to attain the aforementioned objectives are node and edge removal.
Edge removal has garnered greater attention recently~\cite{yao2015minimizing,tong2017efficient,kuhlman2013blocking}, since it is less intrusive (i.e., disrupts the original functionality and flow of the network less aggressively) and provides controlling power in a more granular level (note that usually removing a node is equivalent to removing all its adjacent edges). In this work, we focus on edge removal as well.

Previous studies~\cite{waniekHidingIndividualsCommunities2018,jienhanceprivacy2019,ji2019greedily,lanetprotect2021} have investigated the problem of removing a fixed number of edges to achieve an objective with respect to a subset of target nodes. The objective could vary from minimizing the spreading power of the target nodes under a specific information diffusion model such as the Independent Cascade (with the goal of halting the propagation of a piece of misinformation)~\cite{khalil2014scalable, yao2015minimizing,kuhlman2013blocking,chen2010scalable,yan2019rumor} to minimizing the centrality of the target nodes measured by degree~\cite{waniekHidingIndividualsCommunities2018} or closeness~\cite{ji2019greedily} (with the aim of concealment). Additionally, numerous algorithms, predominantly heuristics, have been put forth~\cite{chen2010scalable,yan2019rumor}.

The existing works have two major limitations. Firstly, despite a plethora of centrality indices proposed in the literature~\cite{LU20161}, these prior works do not consider information centrality~\cite{stephenson1989rethinking, newmanMeasureBetweennessCentrality2005}, in spite of its obvious advantages.  For example, information centrality captures the way node-to-node transmission takes place (especially, on social networks)~\cite{brandesCentralityMeasuresBased2005} and possesses higher discrimination power than others~\cite{shan2018improve,li2019current}. Again for instance, information centrality has been applied to various fields, such as estimation from relative measurements~\cite{BaHe08} and leader selection for noisy consensus~\cite{PoYoScLe16}. 
Secondly, they are computationally relatively expensive which renders them impractical in various real-world scenarios. Suppose a misinformation detection algorithm has spotted a potential rumor (as the one from above about White House) and the goal is to moderate the spreading power of the initiating node 
by temporarily removing some edges (i.e., not exposing the content from some users to some others), then having a very fast mitigation algorithm in place is very crucial. (Note that pre-calculations are not usually a viable solution here, since these graphs undergo constant reformulation). Another example of this would be an internet service provider seeking to control the network traffic in response to a malfunctioning or malevolent router.

In light of the above limitations, our study aims to fill this gap by devising an objective function capturing information centrality and provide an algorithm that can handle networks of million nodes. In our formulation of the problem of moderating critical nodes, the objective is to remove $k$ edges to minimize the information centrality of a target node, while preserving network connectivity. (The constraint of network connectivity ensures the preservation of network's functionality~\cite{fan2020finder,ren2018dismantling}, a consideration also addressed in other works investigating problems related to edge removal~\cite{lanetprotect2021,ji2019greedily,gusrialdi2018distributed}.)

We prove that this problem is NP-complete. Furthermore, while its objective function is monotonically decreasing, it is not supermodular, posing difficulties in solving it using traditional greedy algorithms. However, we still adopt the standard greedy algorithm in our setup since it has been frequently observed to deliver satisfactory results for many non-supermodular problems~\cite{guo2019targeted,bian2017guarantees}. Despite its simplicity and efficacy, it requires execution of matrix inversion operations, incurring a prohibitive computational cost and rendering it unfeasible for large networks. As the first step towards a faster algorithm, we use the random walk-based approximate Schur complement method~\cite{DuGaGoPe19} to present a faster algorithm, called \ApproxiSC.
To speed up the computation even further, we also leverage the sum estimation method~\cite{feigesum2006,sum2022}, which allows us to provide the algorithm \FastSC which runs in nearly linear time in the number of edges, as our main contribution. 

The rest of our paper proceeds as follows. We first introduce some necessary preliminaries related to our work in Section~\ref{sec:prelimi}. Then, we provide an exact formulation of our problem in Section~\ref{sec:profor} and give an overview of our main theoretical and experimental findings and the techniques used in Section~\ref{sec:contr}. We discuss related works in Section~\ref{sec:related}. Then, we study the computational complexity of our problem and prove related properties of the corresponding objective function in Section~\ref{sec:complexity}. In Section~\ref{sec:detergreedy}, we present the deterministic greedy algorithm, followed by the fast greedy algorithms in Section~\ref{sec:fastgreedy}. We report our performance experiments evaluating the efficiency and effectiveness of our algorithms in Section~\ref{sec:experiment} and conclude the paper in Section~\ref{sec:conclusion}. 

\section{Preliminaries}\label{sec:prelimi}
In this section, we introduce some useful notations and tools to facilitate the description of our problem and algorithms.
\subsection{Notations}\label{sec:notation}
We use normal lowercase letters like $a,b,c$ to denote scalars in $\rea$, normal uppercase letters like $A,B,C$ to denote sets, bold lowercase letters like $\aaa, \bb, \cc$ to denote vectors, and bold uppercase letters like $\AA, \BB, \CC$ to denote matrices. Let $\JJ$ be the matrix of appropriate dimensions with all entries being ones. We use $\AA_{[S,F]}$ to denote the submatrix of $\AA$ with row indices in $S$ and column indices in $F$. We write $\AA_{ij}$ to denote the entry at row $i$ and column $j$ of $\AA$ and $\aa_i$ to denote the $i$-th element of vector $\aa$. We use $\AA_{-T}$ to denote the submatrix of $\AA$ obtained from $\AA$ by deleting rows and columns corresponding to elements in set $T$, and use $\aaa_{-T}$ to denote the vector obtained from $\aaa$ by deleting elements in set $T$. An $n \times n$ matrix $\AA$ is positive semi-definite if $\xx^{\top} \AA \xx \geq 0$ holds for all $\xx \in \mathbb{R}^{n}$. For two positive semi-definite matrices $\AA$ and $\BB$, we use $\BB \preceq \AA$ to denote that matrix $\AA-\BB$ is a semi-definite matrix.
Below, we introduce the notion of $\epsilon$-approximation.
\begin{definition}\label{def:eps-appr-mat}
Given two positive semi-definite matrices $\AA$ and $\BB$ and a real number $\epsilon\in (0, 1)$, we say that $\BB$ is an $\epsilon$-approximation of $\AA$ (abbr. $\BB \approx_{\epsilon} \AA$)
\begin{equation*}
   (1-\epsilon)\AA \preceq \BB \preceq (1+\epsilon)\AA.
\end{equation*}
\end{definition}
If $\AA$ and $\BB$ are degenerated to positive scalars $a,b > 0$, $b$ is called an $\eps$-approximation of $a$ (abbr. $b \approx_{\eps} a$) if $(1-\eps)\, a \leq b \leq (1+\eps)\, a$.

\subsection{Graphs and Related Matrices}
Consider a connected undirected graph $\calG = (V,E)$ where $V$ is the set of nodes and $E \subseteq V \times V$ is the set of edges. Let $n = |V|$ and $m = |E|$ denote the number of nodes and the number of edges, respectively. The Laplacian matrix of $\calG$ is the symmetric matrix $\LL = \DD - \AA$, where $\AA$ is the adjacency matrix whose entry $\AA_{ij}=1$ if node $i$ and node $j$ are adjacent, and $\AA_{ij}=0$ otherwise, and $\DD$ is the degree diagonal matrix $\DD=\text{diag}(\dd_1,\cdots,\dd_n)$ where $\dd_i$ is the degree of node $i$. 
We write $\ee_i$ to denote the $i$-th standard basis vector. We fix an arbitrary orientation for all edges in $\calG$, and for each edge $e=(u,v)\in E$, we define $\bb_e = \bb_{uv} = \ee_{u}-\ee_{v}$, where $u$ and $v$ are head and tail of $e$, respectively. Then, $\LL$ can be rewritten as $\LL = \sum\nolimits_{e\in E}\bb_e \bb_e^\top$. 
Matrix $\LL$ is singular and positive semidefinite with its Moore-Penrose pseudoinverse being $\LL^\dag = \kh{\LL +\frac{1}{n}\JJ}^{-1}-\frac{1}{n}\JJ$.  
The transition matrix is $\PP=\DD^{-1}\AA$, which is a row-stochastic matrix.  

For any non-empty node sets $F\subset V$ and $T=V \backslash F$, we can partition the Laplacian matrix $\LL$ into 4 blocks:
\begin{align*}
    \LL:=\left[\begin{array}{ll}
\LL_{[F, F]} & \LL_{[F, T]} \\
\LL_{[T, F]} & \LL_{[T, T]}
\end{array}\right]
\end{align*}
Then, the \textit{Schur complement} of graph $\calG$ onto node set $T$, denoted by $\SC(T)$, is the matrix in closed form as
\begin{align*}
    \SC(T)=\LL_{[T, T]}-\LL_{[T, F]} \LL_{[F, F]}^{-1} \LL_{[F, T]}.
\end{align*}
$\SC(T)$ is a Laplacian matrix of a graph with node set $T$, and we use $\calG(\SC(T))$ to denote the corresponding graph of $\SC(T)$.

\subsection{Information Centrality}\label{sec:information}

Given a graph $\calG=(V,E)$ and two vertices $x,y\in V$, the effective resistance, which is a form of Euclidean distance~\cite{doyle_snell_1984}, $\mathcal{R}^{\calG}_{xy}$ between nodes $x$ and $y$ is defined as $\mathcal{R}^{\calG}_{xy} = \bb_{xy}^\top \LL^\dag \bb_{xy}$. We refer to the maximum value of pairwise effective resistance in a graph as the effective resistance diameter, and denote it by $\phi$. Based on the physical definition of the effective resistance~\cite{doyle_snell_1984}, $\phi$ is less than the diameter of the graph, which is often small in real-life networks~\cite{WaSt98}.

For any node set $T\subset V$, the \textit{Schur complement} onto $T$ can be viewed as a vertex sparsifier that preserves pairwise effective resistance~\cite{DoBu12,DuGaGoPe19}, which means
\begin{align}
    \mathcal{R}_{xy}^{\calG}=\mathcal{R}_{xy}^{\calG(\SC(T))},
\end{align}
holds for any pair of nodes $x,y \in T$.

For a network $\calG = (V,E)$ and a node $\lea \in V$, we use $\mathcal{R}_{\lea}^{\calG}$ to denote the sum of effective resistances between $\lea$ and all nodes in $V\backslash \{\lea\}$ (we will refer to $\mathcal{R}_{\lea}^{\calG}$ as the resistance distance of node $\lea$ throughout the paper), i.e., 
$
    \mathcal{R}_{\lea}^{\calG} = \sum_{u \in V\backslash\{\lea\}} \mathcal{R}_{uv}^{\calG},
$
which can be restated~\cite{BOZZOresistance2013} in a matrix form as
\begin{align}\label{eqR1}
\mathcal{R}_{\lea}^{\calG} = n \LL^\dag_{vv} + \trace{\LL^\dag}.
\end{align}

The information centrality $\calI_{\lea}^{\calG}$ correlates to the resistance distance $\mathcal{R}_{\lea}^{\calG}$~\cite{stephenson1989rethinking,brandesCentralityMeasuresBased2005}, and can be expressed by:

\begin{align}\label{ItoR}
\calI_{\lea}^{\calG} = \frac{n}{\mathcal{R}_\lea^{\calG}}= \frac{n}{n \LL^\dag_{vv} + \trace{\LL^\dag}}
\end{align}
which is defined on connected networks.

We will remove the superscripts of $\mathcal{R}^{\calG}_{xy}$, $\mathcal{R}^{\calG}_{\lea}$ and $\calI_{\lea}^{\calG}$ when $\calG$ is clear from the context.
\subsection{Supermodular Optimization}\label{sec:greedy}

Let $X$ be a finite set, and $2^X$ be the set of all subsets of $X$. Let $f: 2^X \to \mathbb{R}$ be a set function on $X$. Then, $f$ is called monotone decreasing if for any subsets $S \subset H \subset X$, $f(S) > f(H)$ holds. Furthermore, we say function $f$ is supermodular if for any subsets $S \subset H \subset X$ and any element $a \in X \backslash H$, it satisfies $f(S \cup \{a\}) - f(S) \leq f(H \cup \{a\}) - f(H)$.

The standard greedy algorithm has proven to be an effective solution for the cardinality-constrained set function problem in supermodular optimization, with a guaranteed $(1-1/e)$ approximation ratio. However, many crucial set functions do not satisfy the supermodularity requirement. Nevertheless, the greedy algorithm still frequently produces desirable results for a broad range of non-supermodular applications~\cite{bian2017guarantees,guo2019targeted}.

\section{Problem Formulation}\label{sec:profor}

In this section, we formulate the problem of minimizing the information centrality of a target node by removing edges.

Consider a connected unweighted undirected network $\calG=(V,E)$ and a target node $\lea$.
For any edge set $P \subset E$, define $\calG\setminus P=(V,E\backslash P)$ as the remaining graph resulted by removing edges in edge set $P$. Then we define the set function of the information centrality, $\calI_v(P)=\calI_v(\calG\setminus P)$. Similarly, we can define the set function of the sum of the effective resistance $\mathcal{R}_v(P)$. These two definitions are valid whenever the removal of edges in set $P$ maintains the connectivity of the network.

Rayleigh’s monotonicity law~\cite{doyle_snell_1984} asserts that the effective resistance between any pair of nodes will increase when an edge is removed; hence, the resistance distance $\mathcal{R}_v(P)$ is monotonically increasing, and the information centrality $\calI_v(P)$ is monotonically decreasing. Then the following problem arises naturally: How to optimally remove a subset $\del\subset E$ from the network subject to a cardinality constraint $k$ specifying the maximum number of edges that can be removed so that $\calI_\lea$ is minimized while preserving the connectivity of the graph. Mathematically, the information centrality minimization problem can be stated as follows.

\begin{problem}\label{pro:inforMini}
(\uline{Infor}mation \uline{Cen}trality \uline{Min}imization, \InforCenMin) Given a connected undirected network $\calG=(V,E)$, a predefined target node $\lea$, and an integer $k$, we aim to find the edge set $\del\subset E$ with $\sizeof{\del}=k$, so that the information centrality $\calI_v(P)$ is minimized while simultaneously retaining the connectivity of the network. This set optimization problem can be formulated as:

\begin{equation}\label{eq:pro1}
\del^{*} = \argmin_{P \subset E, |P|=k, \calG\setminus P \textrm{ is connected}} \calI_\lea(P).
\end{equation}
\end{problem}

Simply removing edges incident to the target node seems effective to reduce its information centrality. However, removing nonadjacent edges can provide greater benefits in some cases. For example, consider the Dolphin~\cite{RoAh15} network with 62 nodes, 159 edges, and targeting the green node in Fig.~\ref{fig:toy1}. We find that none of the top-10 edges (colored in red) leading to the largest reduction in the target node's information centrality are adjacent to it. This suggests that the scope of edge removal should not be simplified.
\begin{figure}[htbp]
    \centering
    \includegraphics[width=0.7\linewidth]{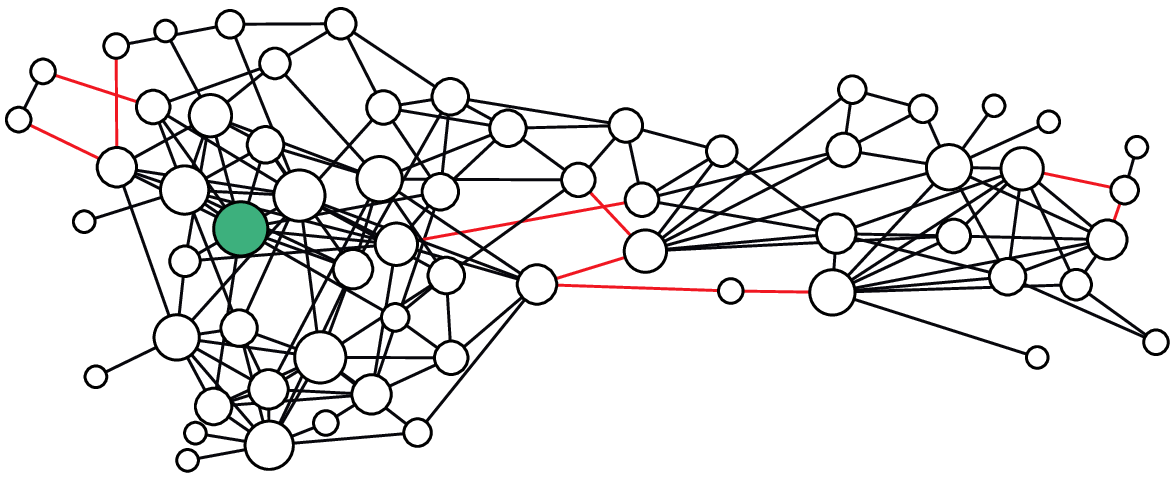}
    \caption{The Dolphin network with the green target node.}
     \label{fig:toy1}
\end{figure}

\section{Our Contribution}
\label{sec:contr}

We investigate the \InforCenMin problem, both theoretically and experimentally, where the goal is to minimize the information centrality of a target node given the possibility to remove $k$ edges while keeping the underlying graph connected. This is a very natural and intuitive formulation of the problem of moderating critical nodes in a network, which in particular takes the very important parameter of information centrality into account.

We first prove that the problem is NP-complete, building on a reduction from the Hamiltonian cycle problem~\cite{karp2010reducibility}. Furthermore, we show that while the objective function of our problem is monotonically decreasing, it does not enjoy supermodularity property, by providing an explicit counterexample. As a result, the traditional greedy approaches would not facilitate us with any theoretical guarantees. However, since such greedy approaches have proven to be a useful base for devising effective and efficient algorithms in the past~\cite{guo2019targeted,bian2017guarantees} (even when super-modularity does not hold) we rely on them as a starting point too.

We first propose a simple algorithm where edges are deleted following the classic greedy approach. This algorithm, called \ExactSM, runs in $O(n^3)$, which makes it impractical for very large networks. As the first step towards a faster algorithm, we use the random walk-based approximate Schur complement method~\cite{DuPePeRa17,DuGaGoPe19} to present a faster algorithm. In this algorithm, called \ApproxiSC, after each edge removal, the new value of the resistance distance and the connectivity status can be determined very quickly by carefully modifying the set of random walks.
To further speed up the computation, as the next step we also leverage the sum estimation method~\cite{feigesum2006,sum2022}, which allows us to provide the algorithm \FastSC that runs in nearly linear time in the number of edges.
The sum estimation method permits us to approximate the resistance distance of the target node by sampling limited pairs of effective resistances. Our theoretical analyses confirm that the combination of the above techniques is well-suited to our problem. Specifically, the absolute error between the approximate resistance distance upon the removal of any edge and the corresponding exact value is at most $\alpha n$, for a small error parameter $\alpha$.

Besides our theoretical analyses, we conduct a comprehensive set of experiments on real-world networks from Network Repository~\cite{RoAh15} and SNAP~\cite{LeSo16} and on synthetic graph models, namely Barabási–Albert (BA)~\cite{albert2002statistical} and Watts–Strogatz (WS)~\cite{WaSt98}. We compare our algorithms against several other algorithms and observe that our algorithms significantly outperform others while producing results very close to optimal solution. Furthermore, our linear time \FastSC algorithm enjoys an extremely fast run time in practice as well. In particular, it completes the task for networks with more than one million nodes in a few hours on a standard 32G-Linux box. Therefore, our algorithms not only allow a rigorous theoretical analysis but also outperform other algorithms in practice, in both aspects of effectiveness and efficiency.

\section{RELATED WORK}\label{sec:related}

In this section, we review the literature related to ours, including minimizing spread of misinformation in social networks, privacy protection of networks, edge removal strategies, and edge centrality measures. Specifically, prior work provided in Sections~\ref{sec:related01} and~\ref{sec:related02} would let us place our contribution in a bigger picture and draw the connection to some adjacent topics, while the results in Sections~\ref{sec:related03} and~\ref{sec:related04} are more closely related.

\subsection{Minimizing Spread of Misinformation}\label{sec:related01}

Consider the setup where a rumor starts spreading in a social network from a known set of nodes and following a predefined spreading dynamics such as the Independent Cascade model~\cite{kempe2003maximizing}. Then, the goal is to contain the spread by blocking $k$ edges~\cite{khalil2014scalable, yao2015minimizing}. A greedy algorithm is developed in~\cite{yao2015minimizing} where in each iteration an edge with maximum containment ability is selected, and some heuristics are proposed in~\cite{khalil2014scalable}. In~\cite{kuhlman2013blocking}, considering the cost for blocking each edge, some budget constraints are defined for critical edge identification. Some heuristic algorithms are then proposed to solve the problem. Applying the maximum influence arborescence method~\cite{chen2010scalable}, an approximation method is proposed in~\cite{yan2019rumor}. These algorithms have several shortcomings. Firstly, they are usually tailored for a particular rumor spreading model such as the Independent Cascade model rather using a more generic notion of information centrality. Secondly, they disregard the important constraint of connectivity, and they might produce disconnected networks. Furthermore, they often fail to cover large networks due to their time complexity.

\subsection{Network Privacy Protection}\label{sec:related02}

One area which is closely related to our work is privacy protection in networks. Most works in this area focus on protecting the privacy of users by structural anonymization. The goal is to modify graph structure to anonymize the underlying network, using various methodologies such as $k$-Anonymity, and differential privacy-based approaches~\cite{karwa2011private,Milanidata2023}. One important objective here is to anonymize the key nodes in a network by reducing their centrality. Previous studies have investigated the problem of removing edges to decrease the centrality of some target node with regard to degree centrality~\cite{waniekHidingIndividualsCommunities2018}, closeness centrality~\cite{ji2019greedily}, and the likelihood that the target node appears in an absorbing random walk~\cite{lanetprotect2021}. However, the notion of information centrality has not been analyzed in this framework.

\subsection{Edge Removal Strategies} \label{sec:related03}

Admittedly, as a practical approach of graph edit, edge removal operation has been extensively used for different application purposes, such as controlling disease spreading~\cite{EnMoBr12,MaKa09}, minimizing the number of spanning trees~\cite{Ra98}, and optimizing eigenvalues of related matrices~\cite{ChToPrElFaFa16,ZhZhCh21}. In social networks, removing edges can correspond to unfriending, not exposing the posts/comments, or maintaining social distance. In computer networks, removing edges is similar to cutting a fiber or bringing down a physical/virtual link temporarily. Many studies on edge removal require the final graph to remain connected. For instance, in~\cite{gusrialdi2018distributed}, the authors have studied the problem of decreasing the greatest eigenvalue of the adjacency matrix by link removal while preserving network connectivity. The authors of~\cite{schoone1987diameter} have investigated the problem of expanding the network diameter by eliminating edges such that the resulting graph remains connected. This is because connectivity is usually essential for the network to preserve its main functionality.

\subsection{Edge Centrality}\label{sec:related04} 

The drop in the information centrality of a node caused by deleting an edge can be used as a measure of its importance. There have been many metrics proposed in the literature to assess the importance of a single edge or a group of edges. Individual edge centrality measures comprise, among others, edge betweenness~\cite{BrPi2007}, spanning edge centrality~\cite{MaCh2015}, and biharmonic distance related edge centrality~\cite{YiSh2018}. Additionally, the importance of an edge can be determined based on the centrality of its endpoints, such as the sum or product of the degree, closeness, and betweenness centrality of the endpoints~\cite{bellingeri2020comparative}. Group edge centrality measures are typically designed to quantify the effect of deleting these edges on specific objective functions, such as 
the inverse geodesic length~\cite{GaNa19}, the total pairwise connectivity~\cite{DiXuThPaZn12}, and the forest index~\cite{ZhBaZh23}. These existing edge centrality measurements are tailored for distinct use cases. Our proposed metric is defined based on the reduction in the information centrality of a target node.

\section{Complexity Challenges}
\label{sec:complexity}

In this section, we study the computational complexity of the \InforCenMin problem. We first consider the decision version of the problem and prove that it is NP-complete in Theorem~\ref{the:NP-com}.
\begin{problem}(\uline{Infor}mation \uline{Cen}trality \uline{Min}imization, \uline{D}ecision Version, \InforCenMinD)
Given a connected undirected graph $\calG=(V, E)$ with $n$ nodes, $m$ edges, and node $\lea$ being the target node, an integer $k \in \mathbb{N}^{+}$, a real number $x \in \mathbb{R}^{+}$, decide whether or not there is a set $\del$ of $k$ edges to be removed from $\calG$ such that $\mathcal{I}_\lea$ is at most $x$ in the connected subgraph $\calG\setminus\del$?
\end{problem}

\begin{theorem}\label{the:NP-com}
    The \InforCenMinD problem is NP-complete.
\end{theorem}

\begin{proof}
We demonstrate a polynomial time reduction from the Hamiltonian cycle problem, which is NP-complete~\cite{karp2010reducibility}. We presume that edge deletion reserves the network's connectivity. Given that we can guess the $k$ edges to be removed and compute $\mathcal{I}_\lea$ in polynomial time, it is apparent that the problem is in NP. The smallest $\mathcal{I}_\lea$ that could possibly be assigned to a node in a connected undirected network with $n$ nodes is $\mathcal{I}_\lea^{\rm{min}}=n/\sum_{i=1}^{n-1}i=\frac{2}{(n-1)(n-2)}$ (e.g., node $\lea$ in Fig.~\ref{fig:toy-nphard}).

\begin{figure}[h]
    \centering
    \includegraphics[width=0.7\linewidth]{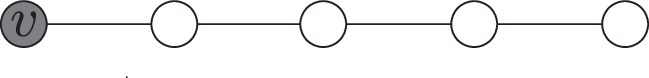}
    \caption{A $5$-node path graph targeting at node $\lea$.\label{fig:toy-nphard}}
\end{figure}

Graph $\calG$ contains a Hamiltonian cycle if and only if $\calG$ has a connected subgraph $\calG^{\prime}$ with $n$ nodes, $n-1$ edges and the information centrality of its end node being $\mathcal{I}_\lea^{\rm{min}}$. So by choosing $k=m-n+1$, and $x=\mathcal{I}_\lea^{\rm{min}}$, we have a reduction from the Hamiltonian cycle problem to the \InforCenMinD problem, proving its NP-completeness.
\end{proof}

A very common technique to tackle NP-hard problems, such as \InforCenMin, is to prove that the objective function enjoys both monotonicity and super-modularity, which consequently would provide us with a Hill Climbing algorithm with a constant approximation guarantee~\cite{guo2019targeted,bian2017guarantees}. While as stated in Lemma~\ref{lemm:mont}, our objective function is monotone, it does not possess super-modularity property, proven in Lemma~\ref{lemm:modular}.

\begin{lemma} (Monotonicity)
\label{lemm:mont}
For two subsets $S$ and $H$ of edges satisfying $S \subset H \subset E$, and $\calG\setminus H$ is a connected graph, we have
\begin{equation*}
    \mathcal{I}_{\lea}(H)<\mathcal{I}_{\lea}(S).
\end{equation*}
\end{lemma}

\begin{lemma} (Non-supermodularity)
\label{lemm:modular}
    $\calI_{\lea}(\cdot)$ is not supermodular.
\end{lemma}

\begin{proof}
To exemplify the non-supermodularity of the objective function~\eqref{eq:pro1}, consider the network in Fig.~\ref{fig:nonsup} (a), a $5$-node graph with node $1$ being the target node and $e_1$ and $e_2$ being edges to delete. We define two edge sets, $S=\emptyset$ and $H=\{e_1\}$. Then, we have
$\calI_{\lea}(S)=1.8$, $\calI_{\lea}({S\cup{\{e_2\}})}=1.4$, 
$\calI_{\lea}({H})=1.3$, and $\calI_{\lea}({H\cup{\{e_2\}})}=0.7$.
Thus, we have
$$\calI_{\lea}(S)-\calI_{\lea}({S\cup{\{e_2\}})}=0.4<0.6=\calI_{\lea}({H})-\calI_{\lea}({H\cup{\{e_2\}})}.$$ This result clearly contradicts the definition of supermodularity. Consequently, the set function of the $\InforCenMin$ problem is not supermodular.
\end{proof}

\begin{figure}[htbp]
    \centering
         \includegraphics[width=0.7\columnwidth]{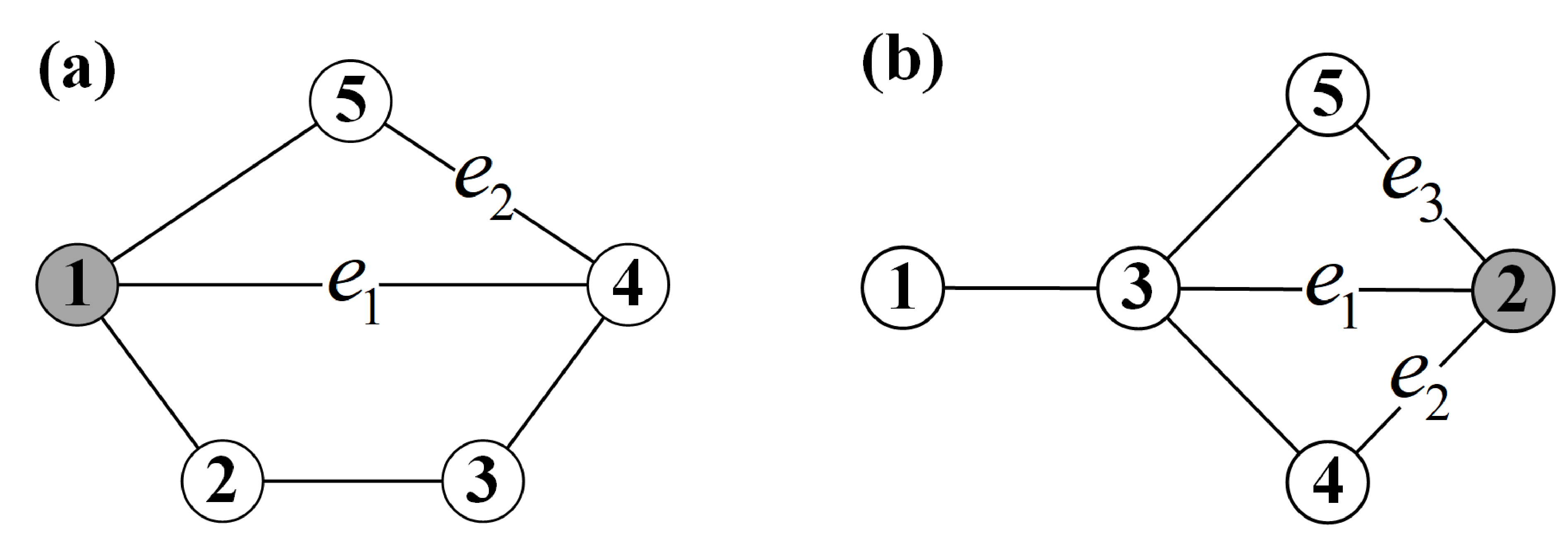}
     \caption{Two $5$-node toy networks.\label{fig:nonsup}}
     \label{fig:toy}
\end{figure}

\section{Deterministic Greedy Algorithm}\label{sec:detergreedy}

The \InforCenMin problem is inherently combinatorial. Its optimal solution can be computed using the following naïve brute-force approach. For each set $\del$ of the $\tbinom{m}{k}$ possible subsets of edges, determine the connectivity of $\calG\setminus \del$, and calculate $\mathcal{I}_\lea(\del)$ in the reduced graph by inverting the Laplacian matrix. Finally, output the subset $\del^{*}$ of $k$ edges whose deletion leads to greatest decrease in $\calI_\lea$ while keeping the connectivity of the graph. Since inverting the Laplacian matrix could take $\Omega(n^3)$ time and there are $\tbinom{m}{k}$ possible subsets, the algorithm's time complexity is in $\Omega\left({m \choose k}n^3\right)$. Thus, albeit its simplicity, this method is computationally unaffordable even for small networks due to its exponential time complexity in $k$.

To tackle this issue, one may consider the heuristic approach of picking the top-$k$ edges with the greatest individual effect on reducing the information centrality of the target node. However, due to the interdependence of edges, the cumulative effect of removing a group of edges is often not equivalent to the sum of individual edge effects. For example, see the network in Fig.~\ref{fig:nonsup} (b), where node 2 is the target node, and edges $e_2$ and $e_3$ are the top-2 edges whose removal has the greatest individual effect on the information centrality of node 2. Surprisingly, removing these top-$2$ edges reduces the information centrality of node 2 to 0.71 while removing edges $e_1$ and $e_2$ would reduce it to $0.56$.

An alternative heuristic is the standard greedy algorithm which starts with an empty set $\del$. Then, in each iteration $i\in \{1,2,\ldots,k\}$, it adds the edge that results in the largest decrease in information centrality of the target node, while preserving connectivity. However, this greedy approach is computationally expensive because of two obstacles present during each iteration: 1) determining whether removal of a certain edge would disconnect the graph could take $\Omega(n)$ time; 2) the computation of the new information centrality through inversion of the matrix might take $\Omega(n^3)$ time. As a result, the total running time could amount to $\Omega(kmn^3)$, which is unaffordable for large networks.

To overcome the first obstacle, some researches focusing on the dynamic connectivity problem resolve connectivity queries in $O(\log n / \log \log \log n)$ time~\cite{wulff2013faster,kapron2013dynamic}. However, these works remain in the realm of theory and can hardly be applied in practice. Moreover, the computational bottleneck caused by 2) is larger; thus, we first focus on reducing the time complexity of updating information centrality.

Let $\mathcal{I}_\lea^\Delta(e) = \mathcal{I}_\lea(\{e\}) - \mathcal{I}_\lea(\emptyset)$ denote the margin gain of the information centrality by removing edge $e$. We provide an efficient method for computing $\mathcal{I}_\lea^\Delta(e)$ in the following lemma.

\begin{lemma}\label{lem:SM}
Let $\calG=(V,E)$ be a connected graph with Laplacian matrix $\LL$. Let $e=(x,y) \in E$ be a candidate edge satisfying that $\calG\setminus\{e\}$ is connected. Then,
\begin{align*}\label{eq:SM}
  \mathcal{I}_\lea^\Delta(e) =\frac{-(nb+n^2c)}{\left(na\LL^{\dag}_{\lea\lea}+nc+a\trace{\LL^{\dag}}+b\right)\left(n\LL^{\dag}_{\lea\lea}+\trace{\LL^{\dag}}\right)},
\end{align*}
where $a= 1-\bb_e^{\top}\LL^{\dag}\bb_e$, $b=\bb_{e}^{\top}\LL^{2\dag}\bb_e$ and $c=\left(\LL^{\dag}\bb_e\right)_{\lea}^{2}$.
\end{lemma}

\begin{proof}
By definition, we have 
\begin{align}
    \notag\mathcal{I}_\lea(\{e\})=\frac{n}{n(\LL-\bb_e\bb_e^\top)^\dag_{\lea\lea}+\trace{(\LL-\bb_e\bb_e^\top)^\dag}}.
\end{align} By Sherman-Morrison formula~\cite{Me73}, we have
\begin{equation}\label{eq:smupdate}
    (\LL-\bb_e\bb_e^\top)^\dag = \LL^\dag+\frac{\LL^\dag\bb_e\bb_e^\top \LL^{\dag}}{1-\bb_e^\top \LL^\dag \bb_e}.
\end{equation}
Note that $\bb_e^\top \LL^\dag \bb_e$ equals the effective resistance between nodes $x$ and $y$, whose value is 1 whenever removing this edge partitions the graph into two components and is less than 1 otherwise.
Thus, the gain in information centrality can be written as
\small 
\begin{align*}
\mathcal{I}_{\lea}^{\Delta(e)} =& \frac{n}{n\big(\LL-\bb_{e}\bb_{e}^{\top}\big)_{\lea\lea}^{\dag}+\trace{(\LL-\bb_{e}\bb_{e}^{\top})^{\dag}}} \\
     &-\frac{n}{n\LL^\dag_{\lea\lea}+\trace{\LL^\dag}} \\
     =&\frac{-(nb+n^2c)}{\left(na\LL^{\dag}_{\lea\lea}+nc+a\trace{\LL^{\dag}}+b\right)\left(n\LL^{\dag}_{\lea\lea}+\trace{\LL^{\dag}}\right)},
\end{align*}
 completing the proof.
\end{proof}

According to Lemma~\ref{lem:SM}, if $\LL^\dag$ is known, we can efficiently compute the marginal gain of the information centrality for one edge by rank-1 update in $O(n)$ time. Then, we propose a deterministic greedy algorithm $\ExactSM(\calG, \lea, k)$. As outlined in Algorithm~\ref{alg:esm}, the first step of this algorithm is to set the result edge set $\del$ to empty and compute $\LL^{\dag}$ in $O(n^3)$ time (Line 1). Then we add $k$ edges to $\del$ iteratively (Lines 2-11). In each iteration, for each candidate edge $e\in E$, we determine the connectivity of the graph $\calG\setminus\{e\}$ in $O(n)$ time (Line 4) and compute the marginal gain of the information centrality in $O(n)$ time (Line 7). After obtaining $\mathcal{I}^\Delta_{\lea}(\cdot)$ for each candidate edge, we select the edge that leads to the smallest marginal gain (Line 8), update the solution (Line 9) and graph (Line 10), and update $\LL^\dag$ according to Equation~\eqref{eq:smupdate} in $O(n^2)$ time. In summary, the total running time of Algorithm~\ref{alg:esm} is $O(n^3+kmn +kn^2)$.

\normalem
\begin{algorithm}
  \caption{$\ExactSM(\calG, \lea , k)$}
  \label{alg:esm}
  \Input{
    A connected graph $\calG=(V,E)$; a target node $\lea \in V$; an integer $k \leq m$
  }
  \Output{
    A subset of $\del \subset E$ with $|\del| = k$
  }
  Set $\del \gets \emptyset$; compute $\LL^\dag$ \;
  \For{$i = 1$ to $k$}{
      \For{$e\in E$}
      {\If {$\calG\setminus\{e\}$ is not connected} {Set $\mathcal{I}^\Delta_{\lea}(e)=0$\;}
      \Else{Compute $\mathcal{I}_{\lea}^\Delta(e)$ by Lemma~\ref{lem:SM}\;}
      }
    Select $e_i$ s.t. $e_i \gets \mathrm{arg\,min}_{e \in E} \mathcal{I}^\Delta_{\lea}(e)$\;
    Update solution $\del \gets \del \cup \{e_i\}$ \;
    Update the graph $\calG \gets (V,E \backslash \{ e_i \})$ \;
    Update $\LL^\dag \gets \LL^\dag+ \frac{\LL^\dag \bb_{e_i} \bb_{e_i}^\top \LL^\dag}{1 -\bb_{e_i}^\top \LL^\dag \bb_{e_i}}$
  }
    \Return $\del$ \;
\end{algorithm}

\section{FAST RANDOMIZED GREEDY ALGORITHM}\label{sec:fastgreedy}

The deterministic greedy algorithm, while faster than the brute-force approach, is not feasible for large networks due to the high computational cost of determining the information centrality marginal gain and graph connectivity. On the positive side, information centrality and the resistance distance are shown to be correlated in Equation~\eqref{ItoR}. This permits us to leverage random walk-based approximate Schur complement method~\cite{DuPePeRa17,DuGaGoPe19} to present a faster algorithm, called \ApproxiSC in Section~\ref{sec:approxisc}. To speed up the computation even further, we then utilize the sum estimation method~\cite{feigesum2006,sum2022}, which allows us to present the algorithm \FastSC in Section~\ref{sec:fastsc} which runs in nearly linear time in the number of edges. 

\subsection{A Simple Sampling Algorithm}
\label{sec:approxisc}
To efficiently calculate and update the information centrality, the main computational bottleneck is the fast calculation of effective resistance Equation~\eqref{ItoR}. Several approximation methods have been proposed to estimate pairwise effective resistance in sublinear time~\cite{DuGaGoPe19,PeLoYoGo21}. However, a na\"ive approach would need to approximate $n$ distinct pairwise effective resistance and then update them for $m$ potential candidate edge. This seems to be computationally expensive. To address this issue, we approach the problem from a different angle, which facilitates us with a much faster mechanism to approximate the effective resistances.

As mentioned in Section~\ref{sec:information}, the Schur complement can be considered as a vertex sparsifier preserving pairwise effective resistance. This motivates us that if we can efficiently compute and update the Schur complement, the aforementioned challenges could be resolved. However, calculating the Schur complement directly is time-consuming. Building upon the ideas of sparsifying random walk polynomials~\cite{ChChLiPeTe15b} and Schur complement~\cite{DuPePeRa17,KyLePeSaSp16}, we approximate it using a collection of random walks that can be maintained upon edge removal with a reasonable time complexity. The following lemma, borrowed from~\cite{DuGaGoPe19,ChChLiPeTe15b, DuPePeRa17,KyLePeSaSp16}, asserts that these walks provide an accurate estimate of the Schur complement.

\begin{lemma}\cite{DuGaGoPe19} \label{lem:appsc}
Let $\calG=(V,E)$ be an undirected unweighted graph with a subset of nodes $T\subset V$. Assume $\epsilon \in (0,1)$, and let $\rho=O\left((\log n)\epsilon^{-2} \right)$ be some sampling concentration parameter. Suppose that $\calH$ is an initially empty graph. For every edge $e=(i,j)\in E$, repeat the following procedure $\rho$ times:
\begin{enumerate}
    \item Simulate a random walk $w_1$ starting from node $i$ until it first hits $T$ at some node $t_1$.
    \item Simulate a random walk $w_2$ starting from node $j$ until it first hits $T$ at some node $t_2$.
    \item Combine these two walks (including $e$) as two sides to get a walk $w=(t_1=u_0, \cdots, u_{\tilde{l}}=t_2)$, where $\tilde{l}$ is the length of the combined walk.
    \item Add the edge $(t_1, t_2)$ to graph $\calH$ with weight $1/\left(\rho \tilde{l} \right)$.
\end{enumerate}

Then, the Laplacian matrix $\LL_{\calH}$ of the resulting graph $\calH$ satisfies $\LL_{\calH} \approx_{\epsilon} \SC(T)$ with probability of at least $1-O(1/n)$. 

\end{lemma}

 Based on Lemma~\ref{lem:appsc}, we can approximate the Schur complement for an arbitrary terminal set $T\subseteq V$ and obtain a graph $\calH$ satisfying $\LL_{\calH} \approx_{\epsilon} \SC(T)$. Let $\tilde{\mathcal{R}}_{xy}$ be the effective resistance between any pair of nodes $x,y\in T$ on graph $\calH$, then based on the fact that matrix approximations also preserve approximations of their quadratic forms, we have
\begin{equation}\label{eq:approx_uv}
    \tilde{\mathcal{R}}_{xy} \approx_\epsilon \mathcal{R}_{xy}^{\calG(\calS(T))} = \mathcal{R}_{xy}^\calG.
\end{equation}

\subsubsection{Approximation of Effective Resistance}
To approximate $\mathcal{R}_{u\lea}$, a direct way is to set $T=\{u,v\}$, and then estimate $\mathcal{R}_{u\lea}$ as the reciprocal of the weight between $u$ and $\lea$ in the resulting graph $\calH$ from Lemma~\ref{lem:appsc}. However, even if the target node $\lea$ is fixed, the approximation of the effective resistance between $v$ and all other nodes is expensive as it may result in redundant walk sampling. To address this issue, we carefully modify the sampling steps in Lemma~\ref{lem:appsc}.

First, we set $T=\{\lea\}$, and sample an initial collection of walks using steps (1)-(3) in Lemma~\ref{lem:appsc}. For each node $u$, we set $T_2=\{u,\lea\}$, and traverse and shorten all the walks at the first position they hit $T_2$, then add the weight of edge $(u,v)$ to the two-node graph $\calH$. This yields a new collection of random walks and an approximation of $\calS(T_2)$. However, repeating this process for every node takes at least $\Omega(mn(\log n) /\epsilon^2)$ time, which is computationally infeasible for large networks. 

We notice that most nodes may appear in a small fraction of the sampled walks; thus, we do not necessarily need to traverse all the walks for each node $u$. So we propose an approach where we traverse each sampled walk exactly once. More precisely, for any walk $w$, we traverse it once and keep track of which nodes appear accompanied by the positions they first appear on both sides. If a node $u$ is only encountered on one side of the walk, setting $T_2=\{u,\lea\}$ will contribute an edge weight to the resulting graph $\calH$. For example, consider the walk in Fig.~\ref{fig:update}(a)-(b)  which starts from the red edge and initially stops at $\lea$. By setting $T_2=\{u, \lea\}$, this walk contributes a weight of $\frac{1}{8\rho}$ to edge $(u,\lea)$ in $\calH$. After summing up the weights contributed by all the walks, we can approximate $\mathcal{R}_{u\lea}$ as the reciprocal of the weight of edge $(u,\lea)$ in $\calH$.

\begin{figure*}[h!]
  \centering
  \includegraphics[width=0.8\linewidth]{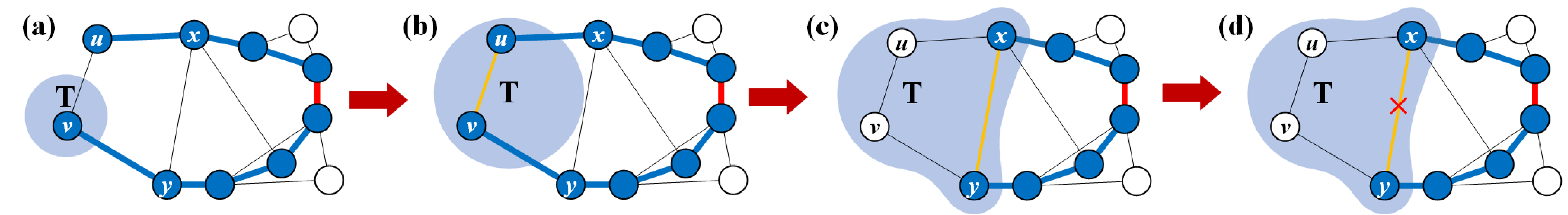}
  \caption{Pipeline of our algorithm. (a) Set $T=\{\lea\}$ and sample an initial walk (colored in blue) starting from the red edge. (b) For an initial approximation of the effective resistance $\mathcal{R}_{u\lea}$, set $T_2=\{u, \lea\}$ and shortcut the initial walk at the first positions it hit $T$ such that it contributes a weight of $\frac{1}{8\rho}$ to edge $(u,\lea)$ (colored in yellow) in the approximate graph. (c) When trying to remove edge $e = (x,y)$, set $T_4=\{u, \lea, x, y\}$ and shortcut the walk at the first positions it hit $T$. (d) After modifying all walks, reduce the weight of edge $(x,y)$ in the approximate graph by 1.}
  \label{fig:update}
\end{figure*}

In summary, we sample random walks according to Lemma~\ref{lem:appsc} by setting $T=\{\lea\}$, and approximate the effective resistances between node $v$ and all nodes $u \in V\backslash \{v\}$ by traversing the sampled walks once. According to Equation~\eqref{eq:approx_uv}, we can approximate the resistance distance $\mathcal{R}_v$ by 
$
    \tilde{\mathcal{R}}_v= \sum_{u\in V\backslash \{v\}} \tilde{\mathcal{R}}_{uv},
$
which satisfies $\sizeof{\mathcal{R}_{\lea} - \tilde{\mathcal{R}}_{\lea}} \leq \varepsilon \mathcal{R}_{\lea} \leq n\epsilon \phi$,
where $\phi$ is the effective resistance diameter of the network~\cite{XuZh23}.

Following Lemma~\ref{lem:appsc}, we need to sample $O(m(\log n)/\epsilon^{2})$ walks. Another critical factor that we need to take into account is the length of sampled walks with an average value of $l_{\text{avg}} = \sum_{u \in V\backslash \{v\}}2\dd_u F_{u,v}/m$, where $F_{u,v}$ represents the hitting time from node $u$ to node $v$. However, some walks may be excessively long, making our computations expensive. To address this issue, we adopt the $l$-\textit{truncated random walk} concept~\cite{SaMoPr08}, where walks are accepted if they shorter than $l$, and disregarded otherwise. Of course, this would result in less accurate solutions. However, we pursue to balance accuracy and efficiency in the choice of $l$. We should stress that this extension is based on the observation that a walk's contribution to a related edge in $\calH$ decreases as its length increases, with a walk longer than $l$ contributing less than $1/(\rho l)$ to the corresponding edge. The following lemma ensures that the expected ratio of invalid walks can be made arbitrarily small for a suitably chosen value of $l$.

\begin{lemma}\cite{zhou2023opinion} \label{lem:maxlen}
Given a connected graph $\calG=(V,E)$, a node set $T=\{v\}$, and a ratio of invalid walks $\gamma>0$, if the maximum length $l$ of random walks satisfies $l= \log(m\gamma /\sqrt{n-1}\norm{\dd_{-T}}_2)/\log(\lambda)$, where $\lambda$ is the spectral radius of matrix $\PP_{-T}$, then the expected ratio of invalid walks is less than $\gamma$.
\end{lemma}

\normalem
\begin{algorithm}[htbp]
  \caption{\textsc{Initialization}$(\calG, \lea, Q, l, \epsilon)$}
  \label{alg:init}
  \Input{
     A connected graph $\calG=(V,E)$; a node $\lea \in V$; a node set $Q$; the maximum length $l$; and a real number $\epsilon \in (0,1)$
  }
  \Output{
    Approximate effective resistance array $\mathcal{R}$
 }
  
Let $\rho = O((\log n)\epsilon^{-2})$; $\calC\gets\mathbf{0}_{n\times 1}$, $\mathcal{R}\gets\mathbf{0}_{n\times 1}$, $T \gets \{\lea\}$ and $W\gets\emptyset$\;
\For{each edge $e=(i,j)\in E$ and $p=1, \cdots, \rho$}{
    Generate a random walk $w_1(e,p)$ from $i$ until it reaches $T$ or its length is $l$\;
    Generate a random walk $w_2(e,p)$ from $j$ until it reaches $T$ or its length is $l$\;
    \If{both walks reach $T$}{
    Combine them to form a walk $w(e,p)$\;
    Add $w(e,p)$ to $W$\;
    }
}
\For{each walk $w \in W$}{
    \For{each node $u \in w\cap Q$}{
        \If{$u$ appears first on one side (A) and not on the other (B)}{
        Let $\tilde{l}$ be the length from $u$ to $\lea$ on $B$ side\;
        Update $\calC\left[u\right]\gets\calC\left[u\right]+1/(\rho \tilde{l})$
        }
    }
}
Let $\mathcal{R}_i \gets 1/\mathcal{C}_i$ for all $i=1,2,\ldots,|\mathcal{C}|$\;
\Return $\mathcal{R}$ and $W$
\end{algorithm}

Based on above analysis, we propose an algorithm \textsc{Initialization} which returns an array $\mathcal{R}$ containing the effective resistances between $\lea$ and all nodes in a given set $Q$ (with $Q=V\backslash\{\lea\}$ in this section), together with a collection of random walks $W$ for future computations. The outline of \textsc{Initialization} is presented in Algorithm~\ref{alg:init}.

\subsubsection{Updating Effective Resistance}
The removal of an edge alters the effective resistances. Recalculating them by resampling walks from scratch is computationally inefficient. To address this issue, we propose an efficient algorithm, called \textsc{DeleteEdge}, which updates effective resistances by modifying the existing walks upon edge removal, as outlined in Algorithm~\ref{alg:delete}. Before discussing the algorithm in more detail, we introduce a node-walk mapping data structure, which maps nodes to the walks that contain them to facilitate the subsequent computation of effective resistances, and is constructed as follows.
\begin{itemize}[leftmargin=*]
    \item First, assign two pristine arrays to each node: the walk array and the position array, aimed at capturing the walks that the node participates in and their corresponding positions, respectively.
    \item Systematically traverse any walk $w\in W$, and scrutinize any node $u$ that is encountered for the first time at position $p$ on either side of the walk. Then append $w$ and $p$ to the walk array and the position array corresponding to $u$.
\end{itemize}
The following example illustrates how this data structure works.
\begin{figure}[h!]
  \centering
  \includegraphics[width=0.9\columnwidth]{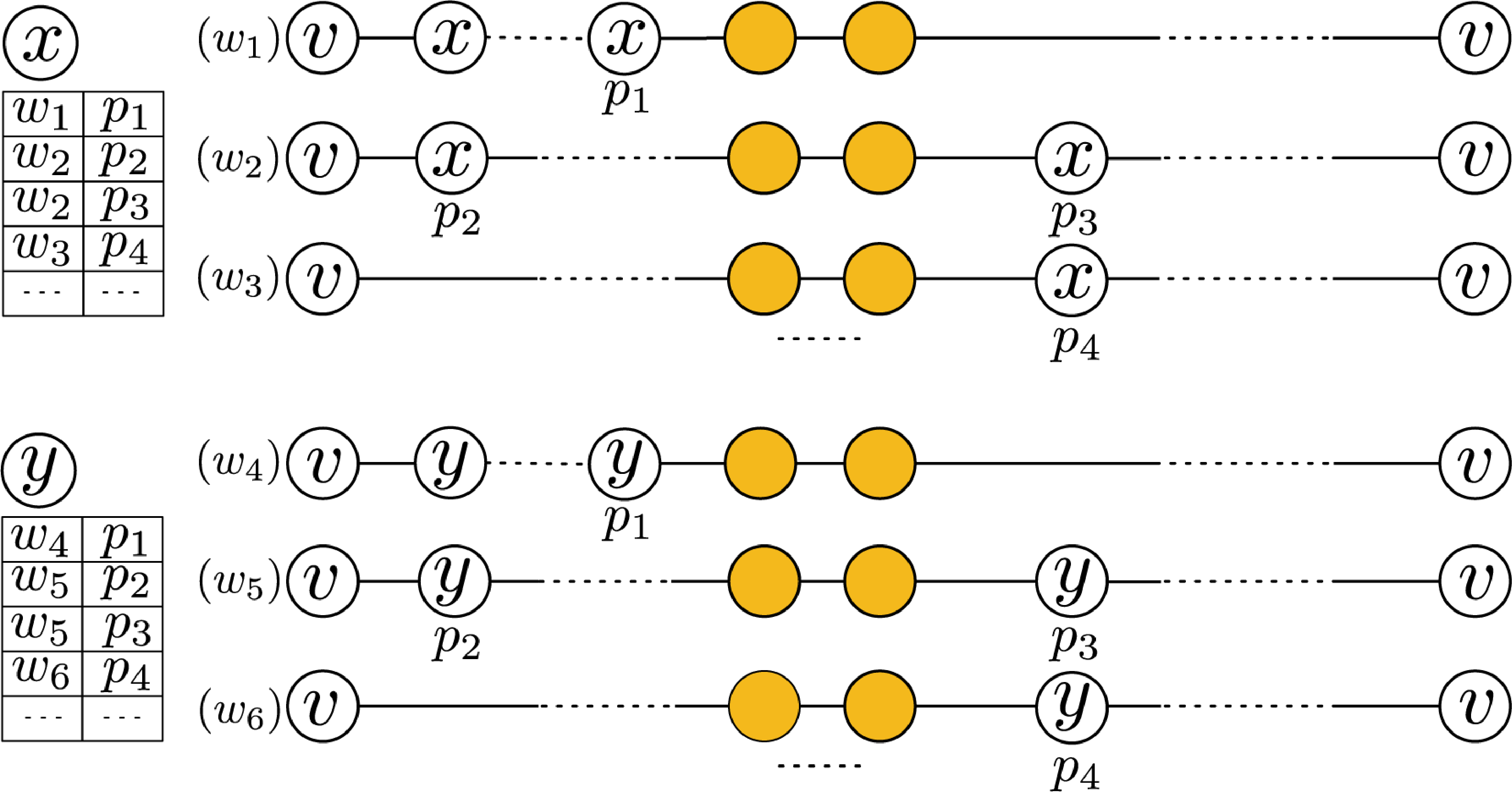}
  \caption{Illustration of the node-walk mapping data structure.}
  \label{fig:node-walk}
\end{figure}

\textbf{Example.} Fig.~\ref{fig:node-walk} demonstrates two instances of the node-walk map. Specifically, node $x$ is successively encountered at positions $p_1$ and $p_2$ on the left
\normalem
\begin{algorithm}[htbp]
  \caption{\textsc{DeleteEdge}$(\calG, W, \lea, Q, \del)$}
    \label{alg:delete}
  \Input{
  A connected graph $\calG=(V,E)$; a collection of random walks $W$; a node $\lea \in V$; a node set $Q$; a set $\del$ of removed edges
     }
     \Output{
     $\{(e, \bar{\mathcal{R}}_{\lea}(\del \cup \{e\}))| e \in E\}$
     }
    \For{each edge $e=(x,y)\in E$}{
        Identify the set of walks $\tilde{W}$ that $x$ or $y$ appears on using the node-walk map\;
        \For{$w \in \tilde{W}$}{
        \For{$u \in w\cap Q$}{
            $T\gets \{u, \lea,x,y\}$\;
            Update the edge weights of $\calH_u$\; 
        }}
            \If{$\calH_x$ \textbf{or} $\calH_y$ is disconnected}{
            Set $\bar{\mathcal{R}}_{\lea}(\del \cup \{e\})=0$\;} 
            \Else{
            \For{$u \in Q$}{
            Update the effective resistance $\tilde{\mathcal{R}}_{u\lea}$\;}
        $\bar{\mathcal{R}}_{\lea}(\del \cup \{e\})\gets\sum_{u\in Q}\tilde{\mathcal{R}}_{u\lea}$\;} 
    }
\Return $\{(e, \bar{\mathcal{R}}_{\lea}(\del\cup \{e\}))| e \in E\}$
\end{algorithm}
side of walk $w_1$, so we append $w_1$ and $p_1$ to the walk array and the position array corresponding to $x$. For walk $w_2$, node $x$ appears on both sides, leading us to record $p_2$ and $p_3$. In turn, for walk $w_3$, we document position $p_4$, where node $x$ is first encountered on the right side. Similar steps can be taken to fill the walk array and position array for node $y$.

Utilizing the node-walk mapping data structure, we present an efficient method for updating the effective resistances upon removal of an edge $e=(x,y)$. The core idea is to reduce the edge removal operation to adding nodes to set $T$ as shown in Fig.~\ref{fig:update} (c)-(d). This reduction is feasible because if nodes $x,y\in T$, then removing edge $e$ in the original graph $\calG$ equates to a decrement in edge $e$'s weight from the Schur complement graph $\calG(\SC(T))$ by $1$. In the following analysis, we fix node $u$. We first sample the random walks with $T=\{v\}$ and then set $T_4=\{u,v,x,y\}$ and denote the resulting approximate Schur complement graph as $\calH_u$. To modify these walks, we utilize the node-walk map to expeditiously determine the walks and positions of nodes $u,x,y$. The corresponding walks are then shortened and the edge weights in graph $\calH_u$ are updated. The final step involves subtracting the edge weight between $x$ and $y$ by $1$. The sufficient condition for the reduced graph being connected is that node $x$ remains accessible from node $y$, which suffices to demonstrate that $\lea$ is reachable from both $x$ and $y$, which can be readily assessed using the updated four-node graph. Consequently, the connectivity of the reduced graph and the new effective resistance between nodes $u$ and $v$ can be easily determined.

Although this method is simple, it still could take $\Omega(mn)$ time for updating all pairs of effective resistances for any edge. However, thanks to the initialization process, we store the information of the Schur complement onto each two-node set $T_2=\{u,v\}$ for each $u\in V\backslash \{v\}$. This would facilitate us to update the four-node graph efficiently. Specifically, we use the node-walk mapping data structure to locate the positions of nodes $x$ and $y$. Then, we shorten walks and update edge weights for all possible four-node sets $T_4=\{q,v,x,y\}$ where $q$ is a node co-occurring with $x$ or $y$ on any walk, and update the edge weights of $\calH_q$. More accurately, for edges $(q,v), (x,v), (y,v)$, the weights are derived from subtracting the weight of $\calG(\mathcal{S}(T_2))$ by the weight reduction due to walk truncation, and for edges $(q,x), (q,y), (x,y)$, the weights are computed by the newly specified walk segments. Finally, we subtract the edge weight between $x$ and $y$ by 1, check the connectivity and compute the effective resistance on the updated four-node graph. We note that the average occurrence of each node in all sampled random walks is $O(mn^{-1}l(\log n)\epsilon^{-2})$, thus the overall time complexity is reduced to $O(m^2n^{-1}l^2(\log n)\epsilon^{-2})$.

To visualize the process of modifying walks, see the walk in Fig.~\ref{fig:update}, which is affected by the removal of edge $(x,y)$. This walk no longer contributes weight to edge $(u,v)$, but contributes weight to edge $(x,y)$.

\subsubsection{A Sampling Algorithm for \InforCenMin}
By integrating the techniques for initializing and modifying random paths to approximate effective resistances, we offer a fast algorithm \ApproxiSC to address the problem \InforCenMin. The pseudocode of this algorithm is summarized in Algorithm \ref{alg:appsc}. Besides the input graph $\calG$, the cardinality constraint $k$ and the target node $\lea$, its input also contains an integer $l$ that constrains the maximum length of random walks and an error parameter $\epsilon$. We first set $\del\gets \emptyset$. Then in each iteration, we initialize the random walks for the current graph by making a call to \textsc{Initialization}. Then, we use \textsc{DeleteEdge} to obtain a set of tuples consisting of an edge and the corresponding $\tilde{\mathcal{R}}_{\lea}$ in the reduced graph. The edge that results in the largest margin gain of $\tilde{\mathcal{R}}_{\lea}$ while preserving graph connectivity will be selected. This process is repeated until $k$ edges have been added to $\del$. 
\normalem
\begin{algorithm}[htbp]
  \caption{$\ApproxiSC(\calG, k, \lea, l, \epsilon)$}
    \label{alg:appsc}
  \Input{
      A connected graph $\calG=(V,E)$; an integer $k<m$; a node $\lea \in V$; the length constraint $l$; an error parameter $\epsilon\in(0,1)$
  }
  \Output{
    An edge set $\del \subset E$ satisfying constraint $k$
  }
Set $\del \gets \emptyset$\;
\For{$i = 1$ to $k$}{
$\mathcal{R} ,W \gets \textsc{Initialization}(\calG, \lea, V, l, \epsilon)$ \;
$\{(e, \tilde{\mathcal{R}}_{\lea}(\del \cup \{e\}))| e \in E\} \leftarrow \textsc{DeleteEdge}(\calG, W, \lea, V, \del)$\;
Select $e_i$ s.t. $e_i \leftarrow \argmax_{e \in E} \tilde{\mathcal{R}}_{\lea}(\del \cup \{e\})$\;
Update solution $\del\gets \del \cup \{e_i\}$\;
Update the graph $\calG \leftarrow (V,E \backslash \{e_i\})$
}
\Return $\del$
\end{algorithm}

We next analyze the time complexity of \ApproxiSC, which consists of two parts: approximating the effective resistances and updating them upon edge removal for all existing edges. In the updating process, we also need to check the connectivity of the resulting graph. Generating the random walks takes $O(ml(\log n) \epsilon^{-2})$ time, initializing the effective resistances takes $O(ml(\log n)\epsilon^{-2})$ time, and updating them takes $O(m^2n^{-1}l^2(\log n) \epsilon^{-2})$ time.
Furthermore, the connectivity can be checked relatively quickly using the derived graph $\calH$. In summary, the overall time complexity of our proposed algorithm \ApproxiSC is in $O(kml(\log n) \epsilon^{-2} + km^2n^{-1}l^2(\log n) \epsilon^{-2})$.
\subsection{A More Efficient Algorithm}
\label{sec:fastsc}

The simple sampling algorithm $\ApproxiSC$ is much faster than the original deterministic algorithm we began with. Nevertheless, in this section we further reduce its time complexity by proposing a faster sampling algorithm, with a similar accuracy guarantee.

\subsubsection{Fast Simulating Random Walks}\label{fast_update}
Each iteration of $\ApproxiSC$ involves simulating $m\lceil(\log n)/\epsilon^2\rceil$ $l$-truncated random walks, which is computationally expensive as $k$ grows.
However, we observe adding an edge $e$ to $\del$ only impacts walks that go through it. Hence, we can speed up the process by only modifying a small fraction of walks when $e$ is deleted, while reusing the rest for subsequent approximations. Next, we elaborate on this more efficient approach.

First, we sample random walks for initialization. Then in each iteration, after selecting an edge $e$, we modify the affected walks. To efficiently locate these walks, we set up an edge-walk mapping data structure. This data structure, similar to the node-walk mapping data structure shown in Fig.~\ref{fig:node-walk}, records the walks and the positions where $e$ first appears, which can be built once the walks are sampled. For each affected walk $w$, we truncate it at the first appearance of edge $e$, and extend it with a new random walk from the truncated position to $\lea$, or until the total length reaches $l$. Both edge-walk and node-walk map, as well as the effective resistance array $\mathcal{R}$, are then updated. Since the average number of walks that include edge $e$ is $O(mn^{-1}l(\log n)\epsilon^{-2})$, updating the walks is efficient.

\subsubsection{Fast Approximation of the Resistance Distance}
\ApproxiSC approximates effective resistances between the target node $\lea$ and all other nodes in the network. However, it is only the sum of these resistances, $\mathcal{R}_\lea$, that is of interest, rather than individual ones. Next, we show that evaluating $\mathcal{R}_\lea$ can be achieved through computing the effective resistances between $\lea$ and a smaller subset $\samv\subseteq V$.

Based on the techniques developed in previous sum estimation works~\cite{feigesum2006,sum2022}, we show that the sum of $n$ bounded elements can be approximated by a sample of them in the following lemma.

\begin{lemma}\label{lem:sqrtsum}
Given $n$ bounded elements $x_1,x_2,\ldots,x_n \in [0,a]$, an error parameter $\beta>an^{-1/2}\log^{1/2} n$, we randomly select $t=O(a\sqrt{n(\log n)}/\beta)$ elements $x_{c_1},x_{c_2},\ldots,x_{c_t}$ by Bernoulli trials with success probability $p=an^{-1/2}\log^{1/2} n/\beta$ satisfying $0<p<1$. We have $\bar{x}=\sum_{i=1}^t x_{c_i}/p$ as an approximation of the sum of the original $n$ elements $x=\sum_{i=1}^n x_i$, satisfying $|x-\bar{x}|\leq n\beta$.
\end{lemma}

Based on the above lemma, let $\phi$ be the effective resistance diameter, $\alpha$ be an error parameter, $\beta=\frac{\alpha}{2}$, and $\epsilon=\frac{\alpha}{2\phi}$. Let $t=O(\phi\sqrt{n(\log n)}/\beta)$, $\tilde{V}=\{x_1,x_2,\ldots,x_t\}$ be a randomly selected subset, and $\hat{\mathcal{R}}_\lea =\sum_{u\in \samv} \tilde{\mathcal{R}}_{u\lea}n/(\phi\sqrt{n(\log n)}/\beta)$, $\tilde{\mathcal{R}}_\lea$ can be approximated by $\hat{\mathcal{R}}_\lea$, satisfying 
\begin{align}
    |\hat{\mathcal{R}}_\lea-\tilde{\mathcal{R}}_\lea|\leq n\beta,
\end{align}
and further $\mathcal{R}_\lea$ can be approximated by $\hat{\mathcal{R}}_\lea$ satisfying
\begin{align}
    |\mathcal{R}_\lea -\hat{\mathcal{R}}_\lea | \leq n\alpha.
\end{align}

By setting $Q=\tilde{V}$, \textsc{Initialization} and \textsc{DeleteEdge} can be simplified by just approximating and updating effective resistances between node $v$ and nodes in set $\tilde{V}$.
\subsubsection{Fast Algorithm for \InforCenMin}
Equipped with the techniques of fast sampling of random walks and efficient computing of the sum of effective resistances, we are now prepared to propose a faster algorithm \FastSC for Problem~\ref{pro:inforMini}, outlined in Algorithm~\ref{alg:fast}. 

\FastSC performs $k$ rounds (Lines 5-9) for iteratively selecting $k$ edges. Given a network $\calG=(V,E)$ with an effective resistance diameter $\phi$, we randomly select a node set $\samv \subset V$ of $O(\phi\sqrt{n(\log n)}/\beta)$ nodes, simulate $m\lceil\phi^2(\log n)/\alpha^2\rceil$ $l$-truncated random walks, and approximate the effective resistances. It takes $O(ml(\log n) \phi^2\alpha^{-2}+ mn^{-1/2}l\log^{3/2} n\alpha^{-3}\phi^3)$ time for the three operations. Then, the algorithm takes $O(km^2n^{-3/2}l^2\log^{3/2} n\alpha^{-3}\phi^3)$ time to update the sum of effective resistances for each candidate edge in all $k$ rounds. Thus, the overall time complexity of \FastSC is $O(ml(\log n) \phi^2\alpha^{-2}$ $+ mn^{-1/2}l\log^{3/2} n\alpha^{-3}\phi^3(1+kml/n))$. Theorem~\ref{the:sqt_performance} characterizes the performance of \FastSC.

\begin{theorem}\label{the:sqt_performance}
    For any $k>0$, an error parameter $\alpha\in(0,1)$, a maximum length $l$ of walks, and the effective resistance diameter $\phi$, \FastSC runs in $O(ml(\log n) \phi^2\alpha^{-2} + mn^{-1/2}l\log^{3/2} n\alpha^{-3}\phi^3(1+kml/n))$ time, and outputs a solution set $\del$ by greedily selecting $k$ edges, such that for the edge selected in each iteration, the information centrality of the target node is maximally decreased. 
\end{theorem}
\normalem
\begin{algorithm}[htbp]
  \caption{$\FastSC(\calG, k, \lea, l, \alpha,\phi)$}
    \label{alg:fast}
  \Input{
      A connected graph $\calG=(V,E)$; an integer $k<m$; a node $\lea \in V$; the maximum length of random walks $l$; a real number $\alpha\in(0,1)$; the effective resistance diameter $\phi$
  }
  \Output{
    An edge set $\del \subset E$ satisfying constraint $k$
  }
Set $\beta=\frac{\alpha}{2}$ and $\epsilon=\frac{\alpha}{2\phi}$; $\del \gets \emptyset$\;
Sample the node set $\samv \subset V$ of $t=O(\phi\sqrt{n(\log n)}/\beta)$\;
$\mathcal{R} ,W \gets \textsc{Initialization}(\calG, \lea, \samv, l, \epsilon)$ \;
\For{$i = 1$ to $k$}{
$\{(e, \hat{\mathcal{R}}_{\lea}(\del \cup \{e\}))| e \in E\} \leftarrow \textsc{DeleteEdge}(\calG, W, \lea, \samv, \del)$\;
Select $e_i$ s.t. $e_i \leftarrow \argmax_{e \in E} \hat{\mathcal{R}}_{\lea}(\del \cup \{e\})$\;
Update solution $\del\gets \del \cup \{e_i\}$\;
Update the graph $\calG \leftarrow (V,E \backslash \{e_i\})$\;
Update the collection of walks $W$ following the method stated in Section~\ref{fast_update}\;
}
\Return $\del$
\end{algorithm}

\section{EXPERIMENT}\label{sec:experiment}
In this section, we evaluate the efficiency and effectiveness of the proposed algorithms through experiments on diverse real-world and synthetic networks of varying types and sizes.

\subsection{Experimental Setup}
In this section, we present the basic experimental settings, which encompass the machine configuration, datasets, baseline algorithms, and choices of parameters. 

\noindent\textbf{Machine Configuration.} Experiments, implemented in \textit{Julia}, are conducted on a Linux server with 32G RAM and 4.2 GHz Intel i7-7700 CPU and with a single thread. The source code is publicly available on \url{https://github.com/hahaabc/fasticm}.

\noindent\textbf{Datasets.} We test our algorithms on both real-world networks, from publically available datasets of Network Repository~\cite{RoAh15} and SNAP~\cite{LeSo16}, and synthetic graphs. The name and some statistics of the experimented real-life networks are presented in Table~\ref{tab:data}, sorted by the number of nodes.

Various synthetic graph models have been introduced to mimic the real-world networks by capturing fundamental properties consistently observed in such networks, such as small diameter and scale-free degree distribution~\cite{albert2002statistical}. We use the well-established and popular BA~\cite{albert2002statistical} and WS~\cite{WaSt98} models. The parameters for these graphs have been chosen such that the average degree closely matches that of real-world networks of comparable size, as in Table~\ref{tab:data}. (All experiments are conducted on the largest component of these graphs due to the connectivity requirement in our problem.) 

\begin{table}
\begin{center}
\caption{Some statistics of the experimented real-world networks. We denote the number of nodes and edges in the largest connected component by $n$ and $m$, respectively, and use $Dim.$ to represent the diameter of a network.}\label{SetNo}
\label{tab:data}
\normalsize
\resizebox{1\columnwidth}{!}{  \setlength{\tabcolsep}{0.8mm} {
\begin{tabular}{cccc|cccc}
\Xhline{2.5\arrayrulewidth}
\raisebox{-0.5ex}{Network} & \raisebox{-0.5ex}{$n$} & \raisebox{-0.5ex}{$m$} & \raisebox{-0.5ex}{$Dim.$} & \raisebox{-0.5ex}{Network} & \raisebox{-0.5ex}{$n$} & \raisebox{-0.5ex}{$m$} & \raisebox{-0.5ex}{$Dim.$}\\[0.5ex]
\hline
\raisebox{-0.5ex}{karate} & \raisebox{-0.5ex}{34} & \raisebox{-0.5ex}{78} & \raisebox{-0.5ex}{5} & \raisebox{-0.5ex}{Erd\"os} & \raisebox{-0.5ex}{6,927}& \raisebox{-0.5ex}{11,850} & \raisebox{-0.5ex}{4}\\[0.5ex]
\raisebox{-0.5ex}{Dolphins} & \raisebox{-0.5ex}{62} & \raisebox{-0.5ex}{159} & \raisebox{-0.5ex}{8}& \raisebox{-0.5ex}{Oregon}& \raisebox{-0.5ex}{10,900} & \raisebox{-0.5ex}{31,180} & \raisebox{-0.5ex}{9}\\[0.5ex]
\raisebox{-0.5ex}{Bomb Train} & \raisebox{-0.5ex}{64} & \raisebox{-0.5ex}{243} &  \raisebox{-0.5ex}{6}& \raisebox{-0.5ex}{ca-HepPh} & \raisebox{-0.5ex}{11,204} & \raisebox{-0.5ex}{117,619}& \raisebox{-0.5ex}{13} \\[0.5ex]
\raisebox{-0.5ex}{Polbooks} & \raisebox{-0.5ex}{105} & \raisebox{-0.5ex}{441} & \raisebox{-0.5ex}{7}& \raisebox{-0.5ex}{Caida} & \raisebox{-0.5ex}{26,475} & \raisebox{-0.5ex}{53,381}& \raisebox{-0.5ex}{17} \\[0.5ex]
\raisebox{-0.5ex}{Hamster} & \raisebox{-0.5ex}{921} & \raisebox{-0.5ex}{4,032} & \raisebox{-0.5ex}{8}& \raisebox{-0.5ex}{Twitter} & \raisebox{-0.5ex}{404,719} & \raisebox{-0.5ex}{713,319}& \raisebox{-0.5ex}{8}\\[0.5ex]
\raisebox{-0.5ex}{Virgili} & \raisebox{-0.5ex}{1,133} & \raisebox{-0.5ex}{5,451}& \raisebox{-0.5ex}{8} & \raisebox{-0.5ex}{Delicious} & \raisebox{-0.5ex}{536,108} & \raisebox{-0.5ex}{1,365,961}& \raisebox{-0.5ex}{14} \\[0.5ex]
\raisebox{-0.5ex}{ca-GrQc} & \raisebox{-0.5ex}{5,242} & \raisebox{-0.5ex}{14,496}& \raisebox{-0.5ex}{17} & \raisebox{-0.5ex}{FourSquare} & \raisebox{-0.5ex}{639,014} & \raisebox{-0.5ex}{3,214,986}& \raisebox{-0.5ex}{4}\\[0.5ex]
\raisebox{-0.5ex}{as20000102} & \raisebox{-0.5ex}{6,474} & \raisebox{-0.5ex}{13,895} & \raisebox{-0.5ex}{9}& \raisebox{-0.5ex}{YoutubeSnap}  & \raisebox{-0.5ex}{1,134,890} & \raisebox{-0.5ex}{2,987,624}& \raisebox{-0.5ex}{20} \\[0.5ex]
\Xhline{2.5\arrayrulewidth}
\end{tabular}
}}
\end{center}
\end{table}

\noindent\textbf{Baselines.}
To better illustrate the effectiveness of our two algorithms, we compare them against the following algorithms. We should stress that all these algorithms are enforced to satisfy the connectivity constraint.

\begin{itemize}[leftmargin=*]
  \item \textsc{Optimum}: find an optimal edge set of size $k$ using brute-force search.
  \item \textsc{Random}: randomly choose $k$ edges.
  \item \textsc{Betweenness}: select $k$ edges with the highest betweenness centrality. The betweenness of an edge here accounts for the number of shortest paths between the target node and other nodes which pass through that edge.
  \item \textsc{Spanning}: select $k$ edges with the highest spanning edge centrality~\cite{MaCh2015}, which is defined as the fraction of spanning trees of a graph that contain a certain edge.
    \item \LaplSolver: pick $k$ edges using Algorithm~\ref{alg:esm} equipped with the approximation technique in~\cite{shan2018improve} plus the connectivity verification process.
\end{itemize}

\noindent\textbf{Parameters.} For algorithms \ApproxiSC and \FastSC, we need to set some parameters. Firstly, smaller values of error parameter $\eps$ in \ApproxiSC and error parameter $\alpha$ in \FastSC provide more accurate results, while larger values result in higher efficiency. We set $\eps=0.005$ and $\alpha=0.05$ in our experiments because as it will be observed in our experiments in Section~\ref{sec:alpha}, they provide a suitable trade-off between accuracy and efficiency. The parameter $\phi$ in \FastSC is the effective resistance diameter of the graph. Since it is inherently difficult to determine its exact value, we instead use the diameter of the graph (reported in Table~\ref{tab:data}), which gives an upper bound on $\phi$. The length $l$ of sampled walks can be determined by Lemma~\ref{lem:maxlen}, which involves the spectral radius of matrix $\PP_{-\{v\}}$, and the ratio of invalid walks $\gamma$. We set the ratio of invalid walks to a relatively small value, namely $0.1\%$. (According to our experiments, $0.1\%$ achieves a good trade-off between effectiveness and efficiency. Otherwise, there is nothing particularly unique about this value.) Computing the spectral radius takes much time for large networks, so we generously estimate it as $\lambda=0.95$.
Unless otherwise specified, the values of the other parameters are set by default when varying a parameter.
\subsection{Effectiveness of Proposed Algorithms}

To evaluate the performance of our algorithms, we compare the quality of the returned edges by the \textsc{Optimum} and \textsc{Random} strategies against ours. The comparison is conducted on four small networks, consisting of two random graphs BA and WS with 50 nodes and two real-world networks, namely Karate and Dolphins. Due to the small size of these networks, we can compare the outcome of our algorithms against the optimal solution (obtained using brute-force search). We select 10 random target nodes. For every network and every $k$ from $1$ to $5$, we run each algorithm 10 times, each time for one of the selected target nodes, and calculate the average final information centrality of target nodes. The lower the obtained value, the better the performance of the algorithm.

The results, depicted in Fig.~\ref{ComOpt}, demonstrate that the performance of the algorithms is consistent across all datasets and that our algorithms perform almost equally to the optimal solutions, reaching parity in the case of the Karate network. On the other hand, the \textsc{Random} scheme consistently underperforms, while the performance of \LaplSolver is similar to that of ours.

\begin{figure}[t]
  \centering
  \includegraphics[width=0.8\columnwidth]{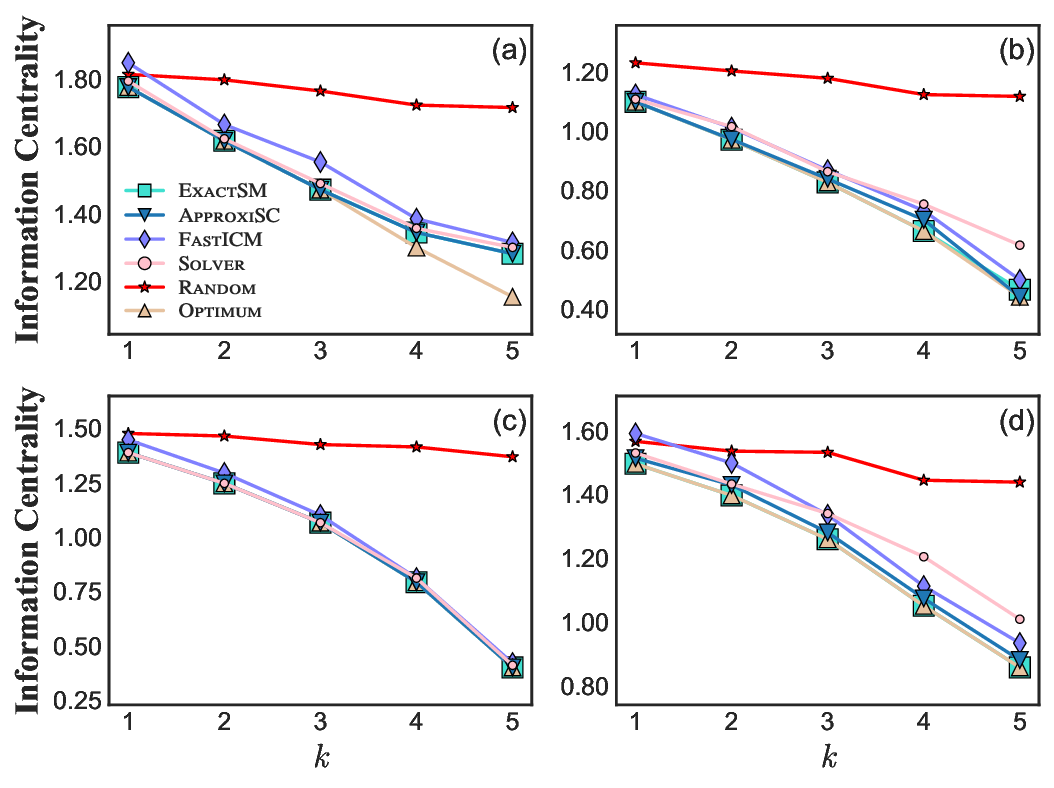} 
  \caption{Average information centrality of target nodes following edge removals, returned by different algorithms on four networks: BA (a), WS (b), Karate (c), and Dolphins (d).\label{ComOpt}}
    \label{fig:comopt}
\end{figure}

As further evidence of the effectiveness of our algorithms, we test them on four larger real-world networks. Since finding optimal solutions through brute-force searching is not feasible due to computational limitations, we compare the results of our algorithms against \textsc{Random}, \textsc{Betweenness}, \textsc{Spanning}, and \LaplSolver algorithms. Similar to above, we randomly select 10 distinct target nodes to mitigate the impact of node positions on the results. We first determine the initial information centrality of each target node and then degrade it by removing up to $k = 10$ edges using our greedy algorithms and four others. After each edge removal, we calculate and record the updated information centrality. Finally, we average over the information centrality of all target nodes in each round and exhibit the results in Fig.~\ref{ComBase}. Our algorithms outperform the other algorithms on all tested networks.
\begin{figure}
\centering
\includegraphics[width=0.8\linewidth]{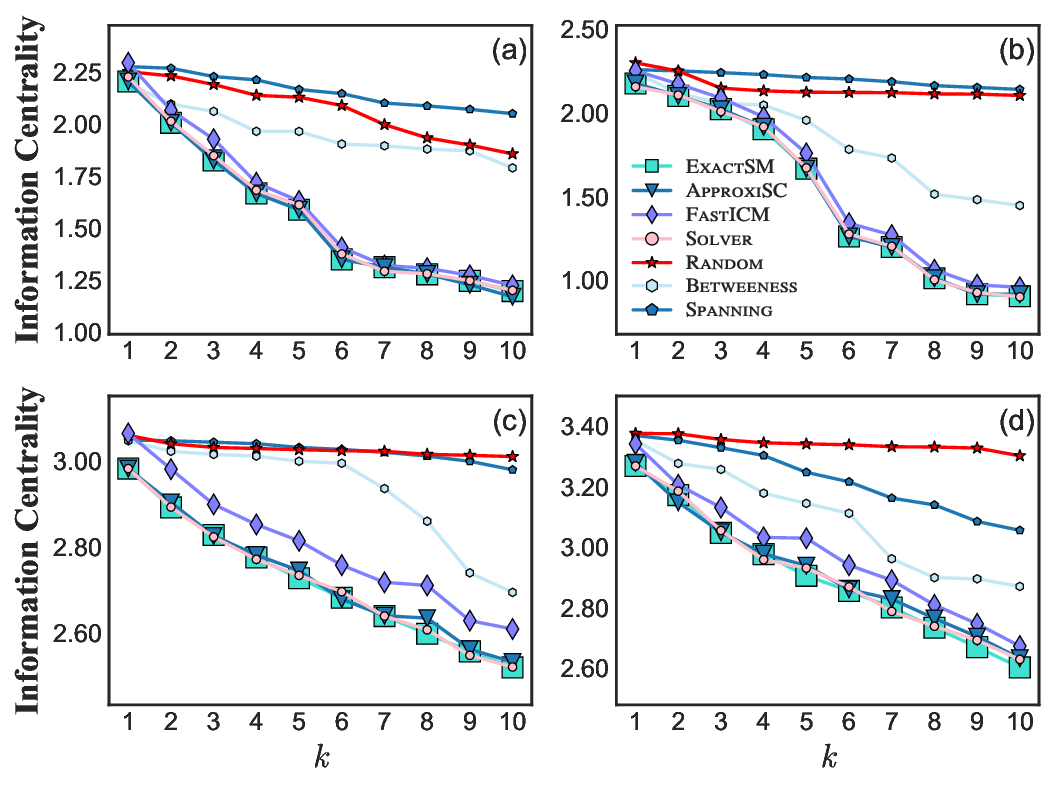}
\caption{Average information centrality of target nodes following edge removals for different algorithms on four networks: Bomb train (a), Virgili (b), ca-GrQc (c), and Erd\"os (d).\label{ComBase} }
\end{figure}
\subsection{Efficiency of Proposed Algorithms}

In the previous section, we observed that our algorithms consistently outperform other algorithms and produce near optimal solutions. Here, we focus on the run time analysis of these algorithms on different networks.

For each network, 20 target nodes are selected randomly, for each of which, $k = 10$ edges are removed to minimize its information centrality using these algorithms. The average value of final information centrality for the targeted nodes and the average run time are reported in Table~\ref{tab:time} for all four algorithms. As expected, the value of information centrality is pretty close for all three algorithms. However, \ApproxiSC and \FastSC are extremely more efficient than the other two algorithms, especially on larger networks. As explained before, \LaplSolver is slower than \ApproxiSC and \FastSC mainly due to its additional connectivity verification cost. \FastSC can handle massive networks, such as FourSquare (with 639,014 nodes) and YoutubeSnap (with 1,134,890 nodes), in only a few hours while \ExactSM and \LaplSolver fail to end before our cut-off time of 24 hours.

\begin{table}[h]
\setlength{\abovecaptionskip}{5.pt}
\setlength{\belowcaptionskip}{-0.cm}

  \caption{The average final information centrality of 20 randomly chosen target nodes and running times for removing $k=10$ edges using $\FastSC$ (FIM), $\ApproxiSC$ (ASC), $\ExactSM$ (ESM) and \LaplSolver (SOL) algorithms on several real-world networks.\label{tab:time}}
  \centering
  \fontsize{7.5}{9}\selectfont
  \begin{threeparttable}
  \setlength{\tabcolsep}{1.2mm}{
    \begin{tabular}{ccccccccc}
    \toprule
    \multirow{2}{*}{Network}&
    \multicolumn{4}{c}{Information Centrality}&\multicolumn{4}{c}{Time (seconds)}\cr
    \cmidrule(lr){2-5} \cmidrule(lr){6-9}
    &FIM&ASC&ESM&SOL&FIM&ASC&ESM&SOL\cr
    \midrule
    Polbooks            &2.667   & 2.661 &2.603 & 2.621 & 3  & 6 & 8 & 20\cr
    Hamster             &2.172   & 2.145 &2.145 & 2.145 & 4  & 6 & 9 & 42\cr
    Virgili             &2.704   & 2.704 & 2.686  & 2.687 & 4  & 9 & 16 & 73\cr
    ca-GrQc             &1.614   & 1.600 & 1.598  & 1.600 & 70  & 83 & 517 & 371\cr
    as20000102          &1.349   & 1.340  & 1.336 & 1.339 & 144  & 194 & 3,175 & 298\cr
    Erd\"os               &1.115   & 1.099 & 1.086 & 1.086 & 281  & 319 & 6,336 & 350\cr
    Oregon              &1.612   & 1.612 & 1.586 & 1.603 & 814  & 912 & 19,357 & 1,568\cr
    ca-HepPh            &2.621   & 2.610 & 2.605 & 2.610 & 412  & 621 & 7,570 & 1,075\cr
    Caida               &1.379   & 1.364  & -    & 1.374   & 2,047   & 2,580 & -      & 4,789  \cr
    Twitter             &1.932   &1.911      &-     &-    & 3,358 & 8,981      &-  &-\cr
    Delicious           &2.012   &1.931       &-  &-       & 7,287 & 25,823      &- &- \cr
    FourSquare          &1.491   &1.401       &- &-       & 19,461 & 68,981      &- &-\cr
    YoutubeSnap         &1.194   &1.111       &- &-      & 11,454 & 57,401      & - &-   \cr
    \bottomrule
    \end{tabular}}
    \end{threeparttable}
    \vspace{-5pt}
\end{table}

\subsection{Impact of Error Parameters}
\label{sec:alpha}
Recall that \ApproxiSC and \FastSC have the error parameters $\eps$ and $\alpha$, respectively. These error parameters balance efficiency and approximation guarantee. We examine their effect on two exemplary datasets, namely Hamster and Virgili. Intuitively speaking, increasing error parameters yields a looser approximation guarantee, which as a result may impact the quality of the selected edge set. Let $\Delta$ be the difference between the final information centrality derived by \ApproxiSC (analogously, \FastSC) and \ExactSM after edge removal. We report in Fig.~\ref{fig:para_alpha} the resulting $\Delta$ after removing 10 edges and the corresponding time consumed for selecting these edges. As expected, smaller $\eps$ (analogously, $\alpha$) yield more accurate results, while larger values demonstrate greater efficiency. The results in Fig.~\ref{fig:para_alpha} should also justify our default choices of $\eps=0.005$ and $\alpha=0.05$ in our experiments since they provide an acceptable balance between the time cost and approximation guarantee. (In these experiments, we have again used 10 randomly chosen target nodes.)

\begin{figure}[h]
\centering
\includegraphics[width=0.8\columnwidth]{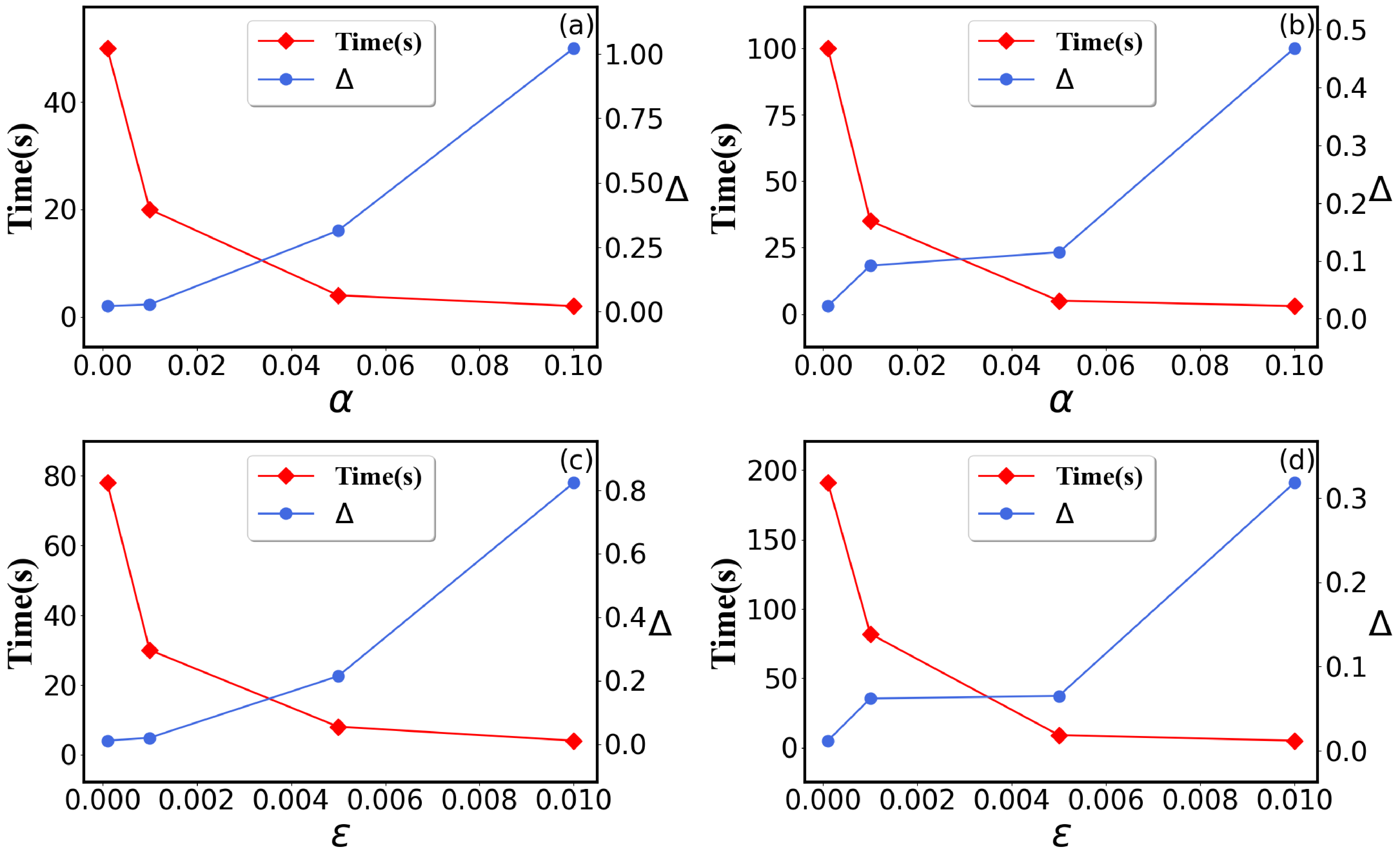}
\caption{Impact of error parameters $\epsilon$ and $\alpha$ in \ApproxiSC and \FastSC algorithms tested on two networks, Hamster(a), and Virgili (b).\label{fig:para_alpha} 
}
\end{figure}

\section{Conclusion and Future Work}\label{sec:conclusion}

Inspired by the imperative of possessing effective mechanisms to mitigate the potential deleterious effects incurred by a compromised/malicious node within a network, we have delved into the problem of moderating critical nodes. We investigated the setup where the goal is to minimize information centrality of a target node by removing $k$ edges while maintaining the network's connectivity. We proved the problem to be NP-complete and its objective function to be monotonically decreasing but non-supermodular. By advancement and development of several novel proof techniques such as random walk-based Schur complement approximation, we provided two fast approximation algorithms and their theoretical analysis. Furthermore, our extensive experiments on various real-world and synthetic networks demonstrated that our proposed algorithms provide solutions very close to optimal in most cases. One of our algorithms, which has a nearly linear run time, can cover networks with over one million nodes on a single machine. Therefore, our algorithms not only permit for a rigorous theoretical analysis, but also perform effectively and efficiently in practice.

We aspire this work to create the necessary grounds for further studies of moderating critical nodes in networks and to be the starting point for a long line of research on this topic. Below, we propose some potential avenues for future research to tackle the limitations of the present work.
\begin{itemize}
    \item \textbf{Connectivity.} We considered the constraint of keeping the underlying network connected while removing edges. Motivated by various real-world applications, it would be important to investigate various notations of connectivity, e.g., $t$-connectivity (for $t\ge 2$), conductance, or algebraic definitions of connectivity.
    \item \textbf{Edge Costs.} In our setup, the cost of removing all edges is the same. However, in the real world, removing some edges might be more costly than the others. Therefore, it would be interesting to investigate the problem when each edge has a cost assigned to it.
    \item \textbf{Multiple Target Nodes.} A natural generalization of the moderation problem is to analysis the setup where one aims to minimize the overall information centrality of several target nodes at once.
    \item \textbf{Weighted Networks.} In practice, edges have weights associated with them. For example, in a social network, an edge weight corresponds to the strength of the social tie between the corresponding two nodes and in an internet network, it could correspond to the bandwidth of the link between two routers. We believe to generalize our algorithms to cover the weighted setup, the introduction of novel ideas and advancement of existing techniques are required.
    \item \textbf{Correlation with diffusion models.} In the future, we will try to find how information centrality correlates with specific diffusion models (cascades, thresholds, etc.)
\end{itemize}


\ifCLASSOPTIONcaptionsoff
  \newpage
\fi

\bibliographystyle{IEEEtran}

\begin{thebibliography}{10}
\providecommand{\url}[1]{#1}
\csname url@samestyle\endcsname
\providecommand{\newblock}{\relax}
\providecommand{\bibinfo}[2]{#2}
\providecommand{\BIBentrySTDinterwordspacing}{\spaceskip=0pt\relax}
\providecommand{\BIBentryALTinterwordstretchfactor}{4}
\providecommand{\BIBentryALTinterwordspacing}{\spaceskip=\fontdimen2\font plus
\BIBentryALTinterwordstretchfactor\fontdimen3\font minus
  \fontdimen4\font\relax}
\providecommand{\BIBforeignlanguage}[2]{{%
\expandafter\ifx\csname l@#1\endcsname\relax
\typeout{** WARNING: IEEEtran.bst: No hyphenation pattern has been}%
\typeout{** loaded for the language `#1'. Using the pattern for}%
\typeout{** the default language instead.}%
\else
\language=\csname l@#1\endcsname
\fi
#2}}
\providecommand{\BIBdecl}{\relax}
\BIBdecl

\bibitem{songced2021}
C.~Song, C.~Yang, H.~Chen, C.~Tu, Z.~Liu, and M.~Sun, ``Ced: Credible early
  detection of social media rumors,'' \emph{IEEE Trans. Knowl. Data Eng.},
  vol.~33, no.~8, pp. 3035--3047, 2021.

\bibitem{zareie2022rumour}
A.~Zareie and R.~Sakellariou, ``Rumour spread minimization in social networks:
  A source-ignorant approach,'' \emph{Online Social Networks and Media},
  vol.~29, p. 100206, 2022.

\bibitem{wang2023lightweight}
Z.~Wang, D.~Hou, C.~Gao, X.~Li, and X.~Li, ``Lightweight source localization
  for large-scale social networks,'' in \emph{Proceedings of the ACM Web
  Conference}.\hskip 1em plus 0.5em minus 0.4em\relax ACM, 2023, pp. 286--294.

\bibitem{freitas2022graphrobustness}
S.~Freitas, D.~Yang, S.~Kumar, H.~Tong, and D.~H. Chau, ``Graph vulnerability
  and robustness: A survey,'' \emph{IEEE Trans. Knowl. Data Eng.}, vol.~35,
  no.~6, pp. 5915--5934, 2023.

\bibitem{bertagnolli2021quantifying}
G.~Bertagnolli, R.~Gallotti, and M.~De~Domenico, ``Quantifying efficient
  information exchange in real network flows,'' \emph{Communications Physics},
  vol.~4, no.~1, p. 125, 2021.

\bibitem{sun2023scalable}
L.~Sun, X.~Rui, and W.~Chen, ``Scalable adversarial attack algorithms on
  influence maximization,'' in \emph{Proceedings of the 16th ACM International
  Conference on Web Search and Data Mining}.\hskip 1em plus 0.5em minus
  0.4em\relax ACM, 2023, pp. 760--768.

\bibitem{ren2018dismantling}
X.-L. Ren, N.~Gleinig, D.~Helbing, and N.~Antulov-Fantulin, ``Generalized
  network dismantling,'' \emph{Proceedings of the National Academy of
  Sciences}, vol. 116, no.~14, pp. 6554--6559, 2019.

\bibitem{YiZhPa20}
Y.~Yi, Z.~Zhang, and S.~Patterson, ``Scale-free loopy structure is resistant to
  noise in consensus dynamics in complex networks,'' \emph{IEEE Trans.
  Cybern.}, vol.~50, no.~1, pp. 190--200, 2020.

\bibitem{shan2018improve}
L.~Shan, Y.~Yi, and Z.~Zhang, ``Improving information centrality of a node in
  complex networks by adding edges,'' in \emph{Proceedings of the 27th
  International Joint Conference on Artificial Intelligence}, 2018, p.
  3535–3541.

\bibitem{EnMoBr12}
E.~A. Enns, J.~J. Mounzer, and M.~L. Brandeau, ``Optimal link removal for
  epidemic mitigation: A two-way partitioning approach,'' \emph{Mathematical
  Biosciences}, vol. 235, no.~2, pp. 138--147, 2012.

\bibitem{MaKa09}
J.~Marcelino and M.~Kaiser, ``Reducing influenza spreading over the airline
  network.'' \emph{PLoS Currents}, vol.~1, 2009.

\bibitem{yan2019rumor}
R.~Yan, Y.~Li, W.~Wu, D.~Li, and Y.~Wang, ``Rumor blocking through online link
  deletion on social networks,'' \emph{ACM Transactions on Knowledge Discovery
  from Data}, vol.~13, no.~2, pp. 1--26, 2019.

\bibitem{tsioutsiouliklis2022link}
S.~Tsioutsiouliklis, E.~Pitoura, K.~Semertzidis, and P.~Tsaparas, ``Link
  recommendations for \text{PageRank} fairness,'' in \emph{Proceedings of the
  Web Conference}.\hskip 1em plus 0.5em minus 0.4em\relax ACM, 2022, pp.
  3541--3551.

\bibitem{LU20161}
L.~Lü, D.~Chen, X.-L. Ren, Q.-M. Zhang, Y.-C. Zhang, and T.~Zhou, ``Vital
  nodes identification in complex networks,'' \emph{Physics Reports}, vol. 650,
  pp. 1--63, 2016.

\bibitem{hofmann2015leadership}
D.~C. Hofmann and O.~Gallupe, ``Leadership protection in drug-trafficking
  networks,'' \emph{Global Crime}, vol.~16, no.~2, pp. 123--138, 2015.

\bibitem{fan2020finder}
C.~Fan, L.~Zeng, Y.~Sun, and Y.-Y. Liu, ``Finding key players in complex
  networks through deep reinforcement learning,'' \emph{Nature Machine
  Intelligence}, vol.~2, no.~6, pp. 317--324, 2020.

\bibitem{peter2013bogus}
F.~Peter, ``bogus’ ap tweet about explosion at the white house wipes billions
  off us markets,'' \emph{The Telegraph}, 2013.

\bibitem{nandi2016methods}
A.~K. Nandi and H.~R. Medal, ``Methods for removing links in a network to
  minimize the spread of infections,'' \emph{Computers \& Operations Research},
  vol.~69, pp. 10--24, 2016.

\bibitem{wang2013negative}
S.~Wang, X.~Zhao, Y.~Chen, Z.~Li, K.~Zhang, and J.~Xia, ``Negative influence
  minimizing by blocking nodes in social networks,'' in \emph{Proceedings of
  the 17th AAAI Conference on Late-Breaking Developments in the Field of
  Artificial Intelligence}, 2013, pp. 134--136.

\bibitem{taninmics2020minimizing}
K.~Tan{\i}nm{\i}{\c{s}}, N.~Aras, {\.I}.~K. Alt{\i}nel, and E.~G{\"u}ney,
  ``Minimizing the misinformation spread in social networks,'' \emph{Iise
  Transactions}, vol.~52, no.~8, pp. 850--863, 2020.

\bibitem{khalil2014scalable}
E.~B. Khalil, B.~Dilkina, and L.~Song, ``Scalable diffusion-aware optimization
  of network topology,'' in \emph{Proceedings of the 20th ACM SIGKDD
  International Conference on Knowledge Discovery and Data Mining}.\hskip 1em
  plus 0.5em minus 0.4em\relax ACM, 2014, pp. 1226--1235.

\bibitem{yao2015minimizing}
Q.~Yao, C.~Zhou, L.~Xiang, Y.~Cao, and L.~Guo, ``Minimizing the negative
  influence by blocking links in social networks,'' in \emph{Trustworthy
  Computing and Services: International Conference, ISCTCS 2014, Beijing,
  China, November 28-29, 2014, Revised Selected papers}.\hskip 1em plus 0.5em
  minus 0.4em\relax Springer, 2015, pp. 65--73.

\bibitem{waniekHidingIndividualsCommunities2018}
M.~Waniek, T.~P. Michalak, M.~J. Wooldridge, and T.~Rahwan, ``Hiding
  individuals and communities in a social network,'' \emph{Nature Human
  Behaviour}, vol.~2, no.~2, pp. 139--147, 2018.

\bibitem{ji2019greedily}
J.~Ji, G.~Wu, C.~Duan, Y.~Ren, and Z.~Wang, ``Greedily remove $k$ links to hide
  important individuals in social network,'' in \emph{International Symposium
  on Security and Privacy in Social Networks and Big Data}, 2019, pp. 223--237.

\bibitem{tong2017efficient}
G.~Tong, W.~Wu, L.~Guo, D.~Li, C.~Liu, B.~Liu, and D.-Z. Du, ``An efficient
  randomized algorithm for rumor blocking in online social networks,''
  \emph{IEEE Transactions on Network Science and Engineering}, vol.~7, no.~2,
  pp. 845--854, 2017.

\bibitem{kuhlman2013blocking}
C.~J. Kuhlman, G.~Tuli, S.~Swarup, M.~V. Marathe, and S.~Ravi, ``Blocking
  simple and complex contagion by edge removal,'' in \emph{2013 IEEE 13th
  International Conference on Data Mining}.\hskip 1em plus 0.5em minus
  0.4em\relax IEEE, 2013, pp. 399--408.

\bibitem{jienhanceprivacy2019}
J.~Ji, G.~Wu, J.~Shuai, Z.~Zhang, Z.~Wang, and Y.~Ren, ``Heuristic approaches
  for enhancing the privacy of the leader in iot networks,'' \emph{Sensors},
  vol.~19, no.~18, 2019.

\bibitem{lanetprotect2021}
R.~Laishram, P.~Hozhabrierdi, J.~Wendt, and S.~Soundarajan, ``Netprotect:
  network perturbations to protect nodes against entry-point attack,'' in
  \emph{Proceedings of the 13th ACM Web Science Conference 2021}, 2021, pp.
  93--101.

\bibitem{chen2010scalable}
W.~Chen, C.~Wang, and Y.~Wang, ``Scalable influence maximization for prevalent
  viral marketing in large-scale social networks,'' in \emph{Proceedings of the
  16th ACM SIGKDD International Conference on Knowledge Discovery and Data
  Mining}.\hskip 1em plus 0.5em minus 0.4em\relax ACM, 2010, pp. 1029--1038.

\bibitem{stephenson1989rethinking}
K.~Stephenson and M.~Zelen, ``Rethinking centrality: Methods and examples,''
  \emph{Social Networks}, vol.~11, no.~1, pp. 1--37, 1989.

\bibitem{newmanMeasureBetweennessCentrality2005}
M.~J. Newman, ``A measure of betweenness centrality based on random walks,''
  \emph{Social Networks}, vol.~27, no.~1, pp. 39--54, 2005.

\bibitem{brandesCentralityMeasuresBased2005}
U.~Brandes and D.~Fleischer, ``Centrality measures based on current flow,'' in
  \emph{Proceedings of the 22nd Annual Conference on Theoretical Aspects of
  Computer Science}, 2005, p. 533–544.

\bibitem{li2019current}
H.~Li, R.~Peng, L.~Shan, Y.~Yi, and Z.~Zhang, ``Current flow group closeness
  centrality for complex networks?'' in \emph{Proceedings of World Wide Web
  Conference}, 2019, pp. 961--971.

\bibitem{BaHe08}
P.~Barooah and J.~P. Hespanha, ``Estimation from relative measurements:
  Electrical analogy and large graphs,'' \emph{IEEE Trans. Signal Process.},
  vol.~56, no.~6, pp. 2181--2193, 2008.

\bibitem{PoYoScLe16}
I.~Poulakakis, G.~F. Young, L.~Scardovi, and N.~E. Leonard, ``Information
  centrality and ordering of nodes for accuracy in noisy decision-making
  networks,'' \emph{IEEE Trans. Autom. Control}, vol.~61, no.~4, pp.
  1040--1045, 2016.

\bibitem{gusrialdi2018distributed}
A.~Gusrialdi, Z.~Qu, and S.~Hirche, ``Distributed link removal using local
  estimation of network topology,'' \emph{IEEE Transactions on Network Science
  and Engineering}, vol.~6, no.~3, pp. 280--292, 2018.

\bibitem{guo2019targeted}
J.~Guo, Y.~Li, and W.~Wu, ``Targeted protection maximization in social
  networks,'' \emph{IEEE Transactions on Network Science and Engineering},
  vol.~7, no.~3, pp. 1645--1655, 2019.

\bibitem{bian2017guarantees}
A.~A. Bian, J.~M. Buhmann, A.~Krause, and S.~Tschiatschek, ``Guarantees for
  greedy maximization of non-submodular functions with applications,'' in
  \emph{Proceedings of the 34th International Conference on Machine Learning},
  2017, pp. 498--507.

\bibitem{DuGaGoPe19}
D.~Durfee, Y.~Gao, G.~Goranci, and R.~Peng, ``Fully dynamic spectral vertex
  sparsifiers and applications,'' in \emph{Proceedings of the 51st Annual ACM
  SIGACT Symposium on Theory of Computing}, 2019, pp. 914--925.

\bibitem{feigesum2006}
U.~Feige, ``On sums of independent random variables with unbounded variance and
  estimating the average degree in a graph,'' \emph{SIAM Journal on Computing},
  vol.~35, no.~4, pp. 964--984, 2006.

\bibitem{sum2022}
L.~Beretta and J.~T{\v{e}}tek, ``Better sum estimation via weighted sampling,''
  in \emph{Proceedings of the Annual ACM-SIAM Symposium on Discrete
  Algorithms}.\hskip 1em plus 0.5em minus 0.4em\relax SIAM, 2022, pp.
  2303--2338.

\bibitem{doyle_snell_1984}
P.~G. Doyle and J.~L. Snell, \emph{Random Walks and Electric Networks}.\hskip
  1em plus 0.5em minus 0.4em\relax Mathematical Association of America, 1984.

\bibitem{WaSt98}
D.~J. Watts and S.~H. Strogatz, ``Collective dynamics of `small-world'
  networks,'' \emph{Nature}, vol. 393, no. 6684, pp. 440--442, 1998.

\bibitem{DoBu12}
F.~Dorfler and F.~Bullo, ``Kron reduction of graphs with applications to
  electrical networks,'' \emph{IEEE Transactions on Circuits and Systems I:
  Regular Papers}, vol.~60, no.~1, pp. 150--163, 2012.

\bibitem{BOZZOresistance2013}
E.~Bozzo and M.~Franceschet, ``Resistance distance, closeness, and
  betweenness,'' \emph{Social Networks}, vol.~35, no.~3, pp. 460--469, 2013.

\bibitem{RoAh15}
R.~Rossi and N.~Ahmed, ``The network data repository with interactive graph
  analytics and visualization,'' in \emph{Proceedings of the AAAI Conference on
  Artificial Intelligence}, 2015, pp. 4292--4293.

\bibitem{karp2010reducibility}
R.~M. Karp, \emph{Reducibility among Combinatorial Problems}.\hskip 1em plus
  0.5em minus 0.4em\relax Springer, 2010.

\bibitem{DuPePeRa17}
D.~Durfee, J.~Peebles, R.~Peng, and A.~B. Rao, ``Determinant-preserving
  sparsification of {SDDM} matrices with applications to counting and sampling
  spanning trees,'' in \emph{IEEE 58th Annual Symposium on Foundations of
  Computer Science}, 2017, pp. 926--937.

\bibitem{LeSo16}
J.~Leskovec and R.~Sosi{\v{c}}, ``{SNAP}: A general-purpose network analysis
  and graph-mining library,'' \emph{ACM Transactions on Intelligent Systems and
  Technolog}, vol.~8, no.~1, p.~1, 2016.

\bibitem{albert2002statistical}
R.~Albert and A.-L. Barab{\'a}si, ``Statistical mechanics of complex
  networks,'' \emph{Reviews of Modern Physics}, vol.~74, no.~1, p.~47, 2002.

\bibitem{kempe2003maximizing}
D.~Kempe, J.~Kleinberg, and {\'E}.~Tardos, ``Maximizing the spread of influence
  through a social network,'' in \emph{Proceedings of the 9th ACM SIGKDD
  International Conference on Knowledge Discovery and Data Mining}.\hskip 1em
  plus 0.5em minus 0.4em\relax ACM, 2003, pp. 137--146.

\bibitem{karwa2011private}
V.~Karwa, S.~Raskhodnikova, A.~Smith, and G.~Yaroslavtsev, ``Private analysis
  of graph structure,'' \emph{Proceedings of the VLDB Endowment}, vol.~4,
  no.~11, pp. 1146--1157, 2011.

\bibitem{Milanidata2023}
M.~Milani, Y.~Huang, and F.~Chiang, ``Data anonymization with diversity
  constraints,'' \emph{IEEE Trans. Knowl. Data Eng.}, vol.~35, no.~4, p.
  3603–3618, apr 2023.

\bibitem{Ra98}
V.~B. Rao, ``Most-vital edge of a graph with respect to spanning trees,''
  \emph{IEEE Transactions on Reliability}, vol.~47, no.~1, pp. 6--7, 1998.

\bibitem{ChToPrElFaFa16}
C.~Chen, H.~Tong, B.~A. Prakash, T.~Eliassi-Rad, M.~Faloutsos, and
  C.~Faloutsos, ``Eigen-optimization on large graphs by edge manipulation,''
  \emph{ACM Transactions on Knowledge Discovery from Data}, vol.~10, no.~4,
  p.~49, 2016.

\bibitem{ZhZhCh21}
Z.~Zhang, Z.~Zhang, and G.~Chen, ``Minimizing spectral radius of
  non-backtracking matrix by edge removal,'' in \emph{Proceedings of the 30th
  ACM International Conference on Information \& Knowledge Management}.\hskip
  1em plus 0.5em minus 0.4em\relax ACM, 2021, pp. 2657--2667.

\bibitem{schoone1987diameter}
A.~A. Schoone, H.~L. Bodlaender, and J.~Van~Leeuwen, ``Diameter increase caused
  by edge deletion,'' \emph{Journal of Graph Theory}, vol.~11, no.~3, pp.
  409--427, 1987.

\bibitem{BrPi2007}
U.~Brandes and C.~Pich, ``Centrality estimation in large networks,''
  \emph{International Journal of Bifurcation and Chaos}, vol.~17, no.~07, pp.
  2303--2318, 2007.

\bibitem{MaCh2015}
C.~Mavroforakis, R.~Garcia-Lebron, I.~Koutis, and E.~Terzi, ``Spanning edge
  centrality: Large-scale computation and applications,'' in \emph{Proceedings
  of the 24th International Conference on World Wide Web}.\hskip 1em plus 0.5em
  minus 0.4em\relax ACM, 2015, pp. 732--742.

\bibitem{YiSh2018}
Y.~Yi, L.~Shan, H.~Li, and Z.~Zhang, ``Biharmonic distance related centrality
  for edges in weighted networks,'' in \emph{Proceedings of the 27th
  International Joint Conference on Artificial Intelligence}, 2018, pp.
  3620--3626.

\bibitem{bellingeri2020comparative}
M.~Bellingeri, D.~Bevacqua, F.~Scotognella, R.~Alfieri, and D.~Cassi, ``A
  comparative analysis of link removal strategies in real complex weighted
  networks,'' \emph{Scientific Reports}, vol.~10, no.~1, pp. 1--15, 2020.

\bibitem{GaNa19}
S.~Gaspers and K.~Najeebullah, ``Optimal surveillance of covert networks by
  minimizing inverse geodesic length,'' in \emph{Proceedings of the AAAI
  Conference on Artificial Intelligence}, vol.~33, 2019, pp. 533--540.

\bibitem{DiXuThPaZn12}
T.~N. Dinh, Y.~Xuan, M.~T. Thai, P.~M. Pardalos, and T.~Znati, ``{On new
  approaches of assessing network vulnerability: Hardness and approximation},''
  \emph{IEEE/ACM Transactions on Networking}, vol.~20, no.~2, pp. 609--619,
  2012.

\bibitem{ZhBaZh23}
L.~Zhu, Q.~Bao, and Z.~Zhang, ``Measures and optimization for robustness and
  vulnerability in disconnected networks,'' \emph{IEEE Trans. Inf. Forensics
  Security}, vol.~18, pp. 3350--3362, 2023.

\bibitem{wulff2013faster}
C.~Wulff-Nilsen, ``Faster deterministic fully-dynamic graph connectivity,'' in
  \emph{Proceedings of the 24th Annual ACM-SIAM Symposium on Discrete
  Algorithms}, 2013, pp. 1757--1769.

\bibitem{kapron2013dynamic}
B.~M. Kapron, V.~King, and B.~Mountjoy, ``Dynamic graph connectivity in
  polylogarithmic worst case time,'' in \emph{Proceedings of the 32nd annual
  ACM symposium on Theory of computing}, 2013, pp. 1131--1142.

\bibitem{Me73}
C.~D. Meyer, Jr, ``Generalized inversion of modified matrices,'' \emph{SIAM
  Journal on Applied Mathematics}, vol.~24, no.~3, pp. 315--323, 1973.

\bibitem{PeLoYoGo21}
P.~Peng, D.~Lopatta, Y.~Yoshida, and G.~Goranci, ``Local algorithms for
  estimating effective resistance,'' in \emph{Proceedings of the 27th ACM
  SIGKDD Conference on Knowledge Discovery and Data Mining}.\hskip 1em plus
  0.5em minus 0.4em\relax ACM, 2021, pp. 1329--1338.

\bibitem{ChChLiPeTe15b}
D.~Cheng, Y.~Cheng, Y.~Liu, R.~Peng, and S.-H. Teng, ``{Efficient sampling for
  Gaussian graphical models via spectral sparsification},'' in \emph{Conference
  on Learning Theory}, 2015, pp. 364--390.

\bibitem{KyLePeSaSp16}
R.~Kyng, Y.~T. Lee, R.~Peng, S.~Sachdeva, and D.~A. Spielman, ``Sparsified
  cholesky and multigrid solvers for connection {L}aplacians,'' in
  \emph{Proceedings of the 48th Annual ACM Symposium on Theory of Computing},
  2016, pp. 842--850.

\bibitem{XuZh23}
W.~Xu and Z.~Zhang, ``Optimal scale-free small-world graphs with minimum
  scaling of cover time,'' \emph{ACM Transactions on Knowledge Discovery from
  Data}, vol.~17, no.~7, p.~93, 2023.

\bibitem{SaMoPr08}
P.~Sarkar, A.~W. Moore, and A.~Prakash, ``Fast incremental proximity search in
  large graphs,'' in \emph{Proceedings of the 25th International Conference on
  Machine Learning}, 2008, pp. 896--903.

\bibitem{zhou2023opinion}
X.~Zhou and Z.~Zhang, ``Opinion maximization in social networks via leader
  selection,'' in \emph{Proceedings of the ACM Web Conference}.\hskip 1em plus
  0.5em minus 0.4em\relax ACM, 2023, p. 133–142.

\end{thebibliography}


\begin{IEEEbiography}
{Changan Liu}
received the B.Eng. degree in School of Software from Dalian University of Technology, Dalian, China, in 2019. He is currently pursuing the Ph.D. degree in School of Computer Science, Fudan University, Shanghai, China. His research interests include network science, computational social science, graph data mining, and social network analysis.
\end{IEEEbiography}
\begin{IEEEbiography}
{Xiaotian Zhou}
received the B.S. degree in mathematics science from Fudan University, Shanghai, China, in 2020. He is currently pursuing the Ph.D. degree in School of Computer Science, Fudan University, Shanghai, China. His research interests include network science, computational social science, graph data mining, and social network analysis.
\end{IEEEbiography}
\begin{IEEEbiography}
{Ahad N. Zehmakan}
received his PhD degree in 2020 from ETH Zurich and is currently an Assistant Professor of Computer Science in the School of Computing at the Australian National University. His research interests include graph and randomized algorithms, information spreading in social networks, random graph models, complexity theory, parallel and distributed computing, and network security.
\end{IEEEbiography}

\begin{IEEEbiography}
{Zhongzhi Zhang}
  (M'19)   received the B.Sc. degree in applied mathematics from Anhui University, Hefei, China, in 1997 and the Ph.D. degree in management science and engineering from Dalian University of Technology, Dalian, China, in 2006. \\
  From 2006 to 2008, he was a Post-Doctoral Research Fellow with Fudan University, Shanghai, China, where he is currently a Full Professor with the School of Computer Science. He has published over 160 papers in international journals or conferences. 
 Since 2019, he has been selected as one of the most cited Chinese researchers
  (Elsevier) every year. 
 His current research interests include network science, graph data mining, social network analysis, computational social science, spectral graph theory, and random walks. \\
  Dr. Zhang was a recipient of the Excellent Doctoral Dissertation Award of Liaoning Province, China, in 2007, the Excellent Post-Doctor Award of Fudan University in 2008, the Shanghai Natural Science Award (third class) in 2013, the CCF Natural Science Award (second class) in 2022, and the Wilkes Award for the best paper published in The Computer Journal in 2019. He is a member of the IEEE.
\end{IEEEbiography}

\end{document}